\documentclass[12pt]{scrartcl}
\usepackage{a4}
\usepackage{amsthm}
\usepackage{amsmath}
\usepackage{amssymb}
\usepackage{amsfonts}
\usepackage{mathrsfs}
\usepackage{latexsym}
\usepackage{dsfont}
\usepackage{color}
\usepackage{bbm,exscale}
\definecolor{Myblue}{rgb}{0,0,0.6}
\usepackage[a4paper,colorlinks,citecolor=Myblue,linkcolor=Myblue,urlcolor=Myblue,pdfpagemode=None]{hyperref}
\usepackage[square,numbers,sort&compress]{natbib}
\usepackage[all,cmtip]{xy}
\usepackage{tikz}

  \vfuzz \hfuzz
\newcommand{\D}{\text{d}}
\newcommand{\E}{\text{e}}
\newcommand{\I}{\text{i}}
\newcommand{\C}{\mathds{C}}

\newcommand{\N}{\mathds{N}}

\newcommand{\R}{\mathds{R}}
\newcommand{\Z}{\mathds{Z}}
\newcommand{\be}{\begin{equation}}
\newcommand{\ee}{\end{equation}}
\newcommand{\bes}{\begin{equation*}}
\newcommand{\ees}{\end{equation*}}
\newcommand{\del}{\partial}
\newcommand{\Coder}{\text{Coder}}
\newcommand{\MC}{\operatorname{\mathcal{MC}}}
\newcommand{\HH}{\operatorname{HH}}
\newcommand{\Jac}{\operatorname{Jac}}
\newcommand{\Res}{\operatorname{Res}}
\newcommand{\MFW}{\operatorname{MF}(W)}
\newcommand{\DGW}{\operatorname{DG}(W)}
\newcommand{\str}{\operatorname{str}}
\newcommand{\id}{\text{id}}

\newcommand{\Hom}{\operatorname{Hom}}
\newcommand{\End}{\operatorname{End}}
\newcommand{\K}[3]{\left#1#3\right#2}

\allowdisplaybreaks
\deffootnote[1em]{1em}{1em}{\textsuperscript{\thefootnotemark}}

\theoremstyle{definition}
\newtheorem{definition}{Definition}
\newtheorem{proposition}[definition]{Proposition}
\newtheorem{theorem}[definition]{Theorem}
\newtheorem{lemma}[definition]{Lemma}
\newtheorem{remark}[definition]{Remark}

\numberwithin{equation}{section}
\numberwithin{definition}{section}
\numberwithin{figure}{section}

\begin{document}

\title{Bulk deformations of open topological string theory}
\author{Nils Carqueville \quad Michael M.~Kay
\\[0.5cm]
\normalsize{\tt \href{mailto:nils.carqueville@physik.uni-muenchen.de}{nils.carqueville@physik.uni-muenchen.de}} \quad \\
 \normalsize{\tt \href{mailto:michael.kay@physik.uni-muenchen.de}{michael.kay@physik.uni-muenchen.de}}\\[0.1cm]
  {\normalsize\slshape Arnold Sommerfeld Center for Theoretical Physics, }\\[-0.1cm]
  {\normalsize\slshape LMU M\"unchen, Theresienstra\ss e~37, D-80333 M\"unchen}\\[-0.1cm]
  {\normalsize\slshape Excellence Cluster Universe, Boltzmannstra\ss e~2, D-85748 Garching}}

\date{}
\maketitle

\vspace{-11.8cm}
\hfill {\scriptsize LMU-ASC 15/11}

\vspace{12cm}

\begin{abstract}
We present a general method to construct bulk-deformed open topological string theories 
from Landau-Ginzburg models. To this end we obtain a weak version of deformation quantisation, and we show how this together with the technique of homological perturbation allows to explicitly compute all bulk-deformed open topological string amplitudes at tree-level before tadpole-cancellation. Our approach is based on a coherent treatment of the problem in terms of the fundamental $A_{\infty}$- and $L_{\infty}$-structures involved. 
\end{abstract}

\newpage

\tableofcontents

\section{Introduction and summary}\label{sec:introduction}

From the worldsheet point of view, perturbative superstring theory is based on $\mathcal N=2$ supersymmetric conformal field theories (CFTs) in two dimensions. Among the central relations of this approach is that string amplitudes are obtained from CFT correlators via integration over the moduli space~$\mathcal M$ of worldsheets, i.\,e.~Riemann surfaces with boundaries and field insertions, see e.\,g.~\cite{GSW,d9703136}. We summarise this relation as the heuristic identity
\be\label{stringCFT}
\big\langle\ldots\big\rangle_{\text{string}} = \int_{\mathcal M} \big\langle\ldots\big\rangle_{\text{CFT}} \, .
\ee

The complete and rigorous treatment of string theory only in terms of CFT is a difficult and not fully solved problem. However, if we restrict to the sector of chiral primary fields~\cite{lvw1989} in $\mathcal N=2$ CFTs, the situation is much better understood. Chiral primaries are known to be equivalently described by the BRST cohomology of the topological twist of the associated $\mathcal N=2$ CFT, see e.\,g.~\cite{w1988,zhoulecturenotes}. This construction is also the most important source of examples of \textsl{open/closed topological field theories (TFTs)}~\cite{a1998,l0010269,ms0609042} which axiomatise properties of the path integral and have a simple and concise algebraic formulation in terms of Frobenius algebras. 

The part of full string theory which is built solely on the structure of a topological conformal field theory (and not a full CFT) is called \textsl{topological string theory}. If we restrict to its tree-level boundary sector then one can rewrite the right-hand side of~\eqref{stringCFT} as
\be\label{bdryamp}
W_{i_{1}\ldots i_{n}} \sim \Big\langle \psi_{i_1}\psi_{i_2}\psi_{i_3}\mathcal P\int\psi_{i_4}^{(1)}\ldots \int\psi_{i_{n}}^{(1)} \Big\rangle_{\text{disk}}
\ee
in terms of chiral primaries $\psi_{i}$ inserted on the boundary and their integrated descendants $\int\psi_i^{(1)}$. A thorough study of BRST symmetry and other Ward identities shows that these topological string amplitudes are subject to certain algebraic constraints. As shown e.\,g.~in~\cite{hll0402}, these relations precisely encode the structure of a \textsl{Calabi-Yau $A_{\infty}$-algebra} on the open string state spaces.\footnote{We shall recall the definitions and their relation to topological string amplitudes $W_{i_{1}\ldots i_{n}}$ in subsection~\ref{subsec:AinfLinf}.}

We may summarise this result together with the alternative approach via the functorial definition of open topological conformal field theory in~\cite{c0412149} as follows: in its complete algebraic formulation, the transition from mere open topological \textsl{field} theory to open topological \textsl{string} theory is the passage from Frobenius algebras to Calabi-Yau $A_{\infty}$-algebras. We thus view the study of such $A_{\infty}$-structures as a first-principle approach to open topological string theory. 

One of the immediate and more concrete consequences of a description of topological string amplitudes $W_{i_{1}\ldots i_{n}}$ via an explicit construction of the $A_{\infty}$-structure is the exact computation of a quantity from full string theory. This is the D-brane superpotential $\mathcal W_{\text{eff}}$ of the associated four-dimensional effective field theory, 
$$
\mathcal W_{\text{eff}}(u) = \sum_{n\geqslant 3} \frac{1}{n} \, W_{i_{1}\ldots i_{n}} \, u_{i_{1}}\ldots u_{i_{n}} \, ,
$$
where $u_{i}$ are the boundary moduli. 

Another important yet non-constructive result of~\cite{hll0402} is that \textsl{bulk deformations} of open topological string theories are described (before tadpole-cancellation) by \textsl{curved} $A_{\infty}$-algebras. On the level of amplitudes, bulk deformations are perturbations of the boundary amplitudes~\eqref{bdryamp} by bulk chiral primary fields $\phi_{j}$ and their integrated descendant $\int\phi^{(2)}_{j}$: 
$$
W_{i_{1}\ldots i_{n}} \longmapsto W_{i_{1}\ldots i_{n}}(t) \sim 
\Big\langle \psi_{i_1}\psi_{i_2}\psi_{i_3}\mathcal P\int\psi_{i_4}^{(1)}\ldots \int\psi_{i_{n}}^{(1)} \, \E^{\sum_{j} t_{j} \int\phi_{j}^{(2)}} \Big\rangle_{\text{disk}} \, .
$$
It is worth stressing that here the general case of perturbations by all chiral primaries $\phi_{j}$ (and not only marginal fields) is considered. 

The appearance of curved $A_{\infty}$-algebras is the main reason for the previous lack of first-principle constructions of bulk-deformed open topological string theories, since the theory of curved $A_{\infty}$-algebras is considerably more complicated and less developed than the non-curved case. Again, if such a construction is successful, it also effortlessly affords the computation of the bulk moduli dependent effective superpotential
$$
\mathcal W_{\text{eff}}(u,t) = \sum_{n\geqslant 1} \frac{1}{n} \, W_{i_{1}\ldots i_{n}}(t) \, u_{i_{1}}\ldots u_{i_{n}} \, .
$$

Alternative approaches to the computation of $\mathcal W_{\text{eff}}(u,t)$ in the case of sigma models with compact targets include those in~\cite{ks0805.1013, ab0909.2245, bbg0704.2666, ahmm0901.2937, mw0709.4028, w0605162, js0808.0761, ahjmms0909.1842, dgjt0203173, gj0608027, ghkk0909.2025,bbs1007.2447,kw0805.0792,gkk1011.6375}. Their successful application usually depends on special properties of the geometry involved, while the $A_{\infty}$-approach is always applicable (though not always most efficient in computations). Conversely, the general approach via $A_{\infty}$-algebras can also be applied to non-compact targets as in~\cite{af0506041,cqv0912.4699}.

\medskip

In the present paper we propose a method based on $A_{\infty}$-theory to construct bulk-deformed open topological string theories from \textsl{Landau-Ginzburg models},%
\footnote{Note that the Landau-Ginzburg models we consider do not necessarily have to correspond to critical string theories.} and we also provide some tools that can be applied in even more general theories. 

Our motivation to study Landau-Ginzburg models is at least twofold. On the one hand, their description as open TFTs via matrix factorisations is comparably simple so that conceptual problems can be addressed unfettered by unnecessary complications. Another motivation has its root in the ultimate goal to understand full string theory rigorously and conceptually from the worldsheet perspective. Among the best-understood classes of CFTs are the so-called rational ones, for which the representation theory of the underlying vertex operator algebra is under comparably good control. However, one may argue that rational CFTs are only ``rather well inspected lampposts''~\cite{d1005.2779}, and there is need for a better grasp on more general theories. A potent tool are Landau-Ginzburg models and their orbifolds, whose infrared fixed points under renormalisation group flow are believed~\cite{kms1989, m1989, vw1989, howewest3} to cover a huge class of $\mathcal N=2$ CFTs. One of the main merits of this conjectured CFT/LG correspondence is the observation that Landau-Ginzburg models discriminate much less against the presence or lack of rational symmetry, and hence they may serve as a guiding light away from the lamppost of rational CFTs. 

\medskip

In the remainder of this introduction we shall give an outline of the main ideas and results of the present paper. The relevant technical background and notions we use will be carefully introduced in section~\ref{sec:mostlyreview}, and section~\ref{sec:bulk-deformedLG} together with the appendix contains the details of our construction. 

We first recall that the on-shell $A_{\infty}$-structure encoding open string amplitudes always originates from an off-shell, string field theoretic description. The latter is encoded in an off-shell state space~$A$ together with the BRST differential and the operator product expansion. They give~$A$ the structure of a differential graded algebra (which is a special case of an $A_{\infty}$-algebra) that we will denote $(A,\del)$. Similarly, the correct on-shell $A_{\infty}$-structure is $(H,\widetilde\del)$ where~$H$ is the BRST cohomology and~$\widetilde\del$ is a short-hand notation for all higher $A_{\infty}$-products. $(H,\widetilde\del)$ is explicitly constructed from $(A,\del)$ as its so-called minimal model (see proposition~\ref{minmodthm} and the discussion in subsection~\ref{oTSTforLG}). Note that the terminology ``on-shell'' (for BRST cohomology~$H$) and ``off-shell'' (for the space~$A$ which also includes non-closed states) that we use throughout the paper is different from the one used in the literature on the computation of $\mathcal W_{\text{eff}}$ by geometrical methods. 
In the latter context ``on-shell'' means that the moduli are chosen such that the F-term equations $\partial \mathcal W_{\text{eff}} = 0$ (which are identical to the Maurer-Cartan equation of the open string $A_\infty$-algebra) are satisfied, and our notion of ``off-shell'' does not play a role. It is, however, the crucial string field theoretic setting from which our approach originates.

In the case of affine Landau-Ginzburg models with potential~$W$, the open string algebras $(A,\del)$ and $(H,\widetilde\del)$ are described in terms of matrix factorisations of~$W$. To study bulk deformations of this open topological string theory, we also need to know the on-shell and off-shell descriptions of the bulk sector. They are given by the Jacobi ring $\Jac(W)$ and the space $T_{\text{poly}}$ of polyvector fields, respectively. It is a general fact that bulk sectors of topological string theories have an $L_{\infty}$-structure, see e.\,g.~\cite{z9206084,z9705241,ks0410291,ks0510118}. For Landau-Ginzburg models this is under good control as we review in subsection~\ref{subsec:LGmodels}. 

To make the connection between the bulk sector and the deformations it induces on the open string algebras $(A,\del)$ and $(H,\widetilde\del)$ more transparent, we must explain the latter's condensed notation in some more detail. (For a more complete discussion we refer to subsection~\ref{subsec:AinfLinf}.) Technically we describe the $A_{\infty}$-structures on~$A$ and~$H$ by nilpotent coderivations~$\del$ and~$\widetilde\del$ on certain spaces ${T_{A}}$ and ${T_{H}}$, respectively. Furthermore, the spaces $\Coder(A,\del)$ and $\Coder(H,\widetilde\del)$ of coderivations on ${T_{A}}$ and ${T_{H}}$ naturally have the structure of $L_{\infty}$-algebras. It turns out that the problem of deforming $(A,\del)$, i.\,e.~finding operators~$\delta$ such that $(A,\del+\delta)$ is still an $A_{\infty}$-algebra, is equivalent to solving the Maurer-Cartan equation for $\Coder(A,\del)$ (and a similar statement is true for~$H$). A general important property is that solution spaces of Maurer-Cartan equations (up to gauge transformations) for quasi-isomorphic $L_{\infty}$-algebras are in bijection, see proposition~\ref{prop:MCsolutions}. 

We are now in a position to present our strategy to construct all bulk deformations of the Landau-Ginzburg open topological string theory described by $(H,\widetilde\del)$. This strategy divides into two main steps. The first step classifies all deformations of the off-shell algebra $(A,\del)$, while the second step transports these deformations to the on-shell algebra $(H,\widetilde\del)$. These are our main results, and they are the content of theorems~\ref{Th1} and~\ref{thm:bulkdeformations}, respectively. 
Our construction is the natural stringy bulk-boundary map that provides
the bulk-deformed open topological string amplitudes before tadpole-cancellation, as in~\cite{hll0402}.

To take the first step we will construct an $L_{\infty}$-quasi-isomorphism
\be\label{TC}
T_{\text{poly}} \longrightarrow \Coder(A,\del)
\ee
which links the off-shell bulk space $T_{\text{poly}}$ with the coderivations governing deformations of the off-shell open string algebra $(A,\del)$. It will be easy to see that the solutions of the Maurer-Cartan equation for $T_{\text{poly}}$ are precisely given by the on-shell bulk space $\Jac(W)$, and since~\eqref{TC} is a quasi-isomorphism we thus know that all deformations of $(A,\del)$ come from (on-shell) bulk fields. Furthermore, as an aside from the main subject of the paper, we will also sketch a construction of an off-shell enhancement of the Kapustin-Li pairing~\cite{kl0305,hl0404} that appears naturally in the context of our first step. 

To construct the map~\eqref{TC} in subsection~\ref{WDQ}, we first translate the problem of deforming the differential graded algebra $(A,\del)$ into that of a curved algebra $(A,\del')$ (with its curvature originating from the potential~$W$) by finding an isomorphism $\Coder(A,\del)\cong \Coder(A,\del')$. Then we observe that the $L_{\infty}$-quasi-isomorphism
$$
T_{\text{poly}} \longrightarrow \Coder(A,\del')
$$
is known in the special commutative case in which also both the curvature in $\del'$ and the potential~$W$ (as part of the $L_{\infty}$-structure of $T_{\text{poly}}$) vanish. To wit, this case is precisely the construction of Kontsevich's \textsl{deformation quantisation}~\cite{k9709040,t9803025,cf9902090}. With this deep result as a given, we show how to generalise it to our case of interest, thus completing our first main step in the proof of theorem~\ref{Th1}. Since curved $A_{\infty}$-algebras are sometimes also called weak, we refer to this construction as ``weak deformation quantisation''. 

In the second step we have to understand how the bulk deformations $(A,\del+\delta)$ of the off-shell algebra are passed down to the on-shell algebra $(H,\widetilde\del)$. To this effect we note that the method of the \textsl{homological perturbation lemma}~\cite{GuLaSt,c0403266} can be employed in our setting. As is explained together with the rest of our second step in subsection~\ref{transbulkonshell}, under certain conditions the homological perturbation lemma allows to transport deformations of one complex to another one. We show that these conditions are met for the complexes $({T_{A}},\del)$ and $({T_{H}},\widetilde\del)$, and we thus obtain an on-shell deformation $\widetilde\del+\widetilde\delta$ from $\del+\delta$. We emphasise that this construction is entirely explicit and works for arbitrary $A_{\infty}$-algebras. 

From the homological perturbation lemma alone one however cannot decide whether or not $(\widetilde\del+\widetilde\delta)\in\End({T_{H}})$ actually is a coderivation and thus really encodes an $A_{\infty}$-structure. Yet by reformulating homological perturbation in coherent $L_{\infty}$-language (proposition~\ref{HPL}) we construct another $L_{\infty}$-morphism
$$
\Coder(A,\del) \longrightarrow \Coder(H,\widetilde\del)
$$
and thus find that our deformations $\widetilde\del+\widetilde\delta$ are indeed coderivations (proposition~\ref{HPLCoder}). In this way we construct all bulk-induced deformations of Landau-Ginzburg open topological string theories from first principles.

\section{Review of the relevant notions and results}\label{sec:mostlyreview}

In this section we first introduce basic notions and results from the theory of $A_{\infty}$- and $L_{\infty}$-algebras. Then we review B-twisted Landau-Ginzburg models by recalling how they are  known to be endowed with the structures of topological field theory and open topological string theory. 

\subsection[$A_{\infty}$- and $L_{\infty}$-algebras]{$\boldsymbol{A_{\infty}}$- and $\boldsymbol{L_{\infty}}$-algebras}
\label{subsec:AinfLinf}

A \textsl{curved $A_{\infty}$-algebra} is a ($\Z$- or $\Z_{2}$-) graded vector space\footnote{We always work over the field $\C$.} $A=\bigoplus_{i} A_{i}$ together with  a codifferential $\del$ of degree $+1$ on the tensor coalgebra
$$
{T_{A}} = \bigoplus_{n\geqslant 0} A[1]^{\otimes n} \, .
$$
To unwrap this definition we first recall that the suspended vector space $A[1]$ has homogeneous components $A[1]_{i}=A_{i+1}$, and that the standard comultiplication $\Delta: {T_{A}} \rightarrow {T_{A}} \otimes {T_{A}}$ is given by
$$
\Delta (a_{1}\otimes\ldots\otimes a_{n}) = \sum_{j=0}^n (a_{1}\otimes\ldots\otimes a_{j}) \otimes (a_{j+1}\otimes\ldots\otimes a_{n}) \, .
$$
Furthermore, a \textsl{coderivation} $\del\in\text{Coder}({T_{A}})$ is a linear operator on ${T_{A}}$ that satisfies the dual version of the product rule, 
\be\label{coLeibniz}
\Delta \circ \del = \left( \del \otimes \id_{{T_{A}}} + \id_{{T_{A}}} \otimes \del \right) \circ \Delta \, ,
\ee
and a \textsl{codifferential} is a coderivation that squares to zero, 
\be\label{del2}
\del^2 = 0 \, .
\ee

If we decompose the codifferential as
$$
\del = \sum_{m,n\geqslant 0} \del_{m}^n \quad\text{where}\quad \del_{m}^n \in \Hom(A[1]^{\otimes m},A[1]^{\otimes n}) \, ,
$$
then the maps
$$
r_{n} = \del_{n}^1: A[1]^{\otimes n} \longrightarrow A[1]
$$
are fundamental since it follows from~\eqref{coLeibniz} that all the other components of~$\del$ can be expressed in term of them:
\be\label{manytomany}
\del_{m}^n=\sum_{j=0}
^{n-1} \id_{A[1]}^{\otimes j} \otimes \del_{m-n+1}^1 \otimes \id_{A[1]}^{\otimes (n-j-1)} \, .
\ee
Thus the $A_{\infty}$-structure encoded in~$\del$ is equivalently described by the maps~$r_{m}$, for which the condition~\eqref{del2} translates into the bilinear constraints
\be\label{Ainfconst}
\sum_{i,j\geqslant 0, \, i+j\leqslant n} r_{n-j+1} \circ \left( \id_{A[1]}^{\otimes i} \otimes r_{j} \otimes \id_{A[1]}^{\otimes (n-i-j)} \right) = 0 
\ee
for all $n\geqslant 0$. We write $\del_{n}$ for the codifferential determined solely by $r_{n}$, and we have the decomposition $\del = \sum_{n\geqslant 0}\del_{n}$. 

A curved $A_{\infty}$-algebra $(A,r_{n})$ is said to be \textsl{cyclic} with respect to a pairing $\langle\,\cdot\,,\,\cdot\,\rangle: A\otimes A\rightarrow \C$ if
$$
\langle a_{0}, r_{n}(a_{1}\otimes\ldots\otimes a_{n})\rangle = (-1)^{\widetilde a_{0}(\widetilde a_{1}+\ldots+\widetilde a_{n})} \langle a_{n}, r_{n}(a_{2}\otimes\ldots\otimes a_{n}\otimes a_{0})\rangle
$$
for all $n\geqslant 0$ and all homogeneous $a_{i}$, where we write $\widetilde a$ for the degree of $a\in A[1]$. $(A,r_{n})$ is \textsl{unital} if there exists an element $e\in A_{0}$ such that $r_{2}(e\otimes a)=-a$, $r_{2}(a\otimes e)=(-1)^{\widetilde a} a$ for all homogeneous $a\in A[1]$, and all other products $r_{n}$ vanish if applied to a tensor product involving~$e$. 

A curved $A_{\infty}$-algebra $(A,r_{n})$ is called a \textsl{strong $A_{\infty}$-algebra} (or simply an \textsl{$A_{\infty}$-algebra}) if its \textsl{curvature} $r_{0}:\C\rightarrow A[1]$ vanishes. In this case~\eqref{Ainfconst} implies that $r_{1}$ is a differential, and $(A,r_{n})$ is called \textsl{minimal} if $r_{1}=0$. A minimal $A_{\infty}$-algebra $(A,r_{n})$ that is unital and cyclic with respect to a non-degenerate pairing is called \textsl{Calabi-Yau}.

\begin{remark}\label{oTSTviaAinf}
Calabi-Yau $A_{\infty}$-algebras appear in open topological string theory in the following way. Let~$H$ be the space of open string operators of the underlying TFT, and let $\{\psi_{i}\}$ be a basis. $H$~is equipped with a non-degenerate pairing which is the two-point-correlator $\langle\,\cdot\,,\,\cdot\,\rangle_{\text{disk}}$. By the results of \cite{hll0402,c0412149}, there exists a minimal and unital $A_{\infty}$-structure $(H,\widetilde r_{n})$ that is cyclic with respect to $\langle\,\cdot\,,\,\cdot\,\rangle_{\text{disk}}$, and the open string amplitudes $W_{i_{1}\ldots i_{n}}$ introduced in~\eqref{bdryamp} are encoded in this $A_{\infty}$-structure via the identity
$$
W_{i_{1}\ldots i_{n}} = \left\langle \psi_{i_{1}} , \widetilde r_{n-1}(\psi_{i_{2}}\otimes\ldots\otimes \psi_{i_{n}}) \right\rangle_{\text{disk}} \, .
$$
The maps $\widetilde r_{n}$ can be explicitly constructed using the minimal model theorem, see proposition~\ref{minmodthm} and the discussion in subsection~\ref{oTSTforLG} below. 
\end{remark}

Another special case of a curved $A_{\infty}$-algebra which is of particular interest is the one whose higher maps $r_{n}$ for $n\geqslant 3$ are all zero, and it is called a \textsl{curved differential graded (DG) algebra}. If we define
$$
C = r_{0}(1) \, ,\quad d=r_{1} \, , \quad a\cdot b = (-1)^{\widetilde a} r_{2}(a\otimes b)
$$
then in this case the constraints~\eqref{Ainfconst} become
\begin{align*}
& d(C) = 0 \, \quad d^2(a) = a\cdot C - C\cdot a \, , \\
& d(a\cdot b) = d(a)\cdot b + (-1)^{|a|}a\cdot d(b) \, , \quad (a\cdot b) \cdot c = a\cdot(b\cdot c)
\end{align*}
for all homogeneous $a,b,c\in A$, and the signs involving the degree $|a| = \widetilde a + 1$ of $a\in A$ arise from the Koszul rule. If the curvature~$C$ vanishes we have a DG algebra, and if both~$C$ and~$d$ vanish we are left with a graded associative algebra. 

\medskip

$A_{\infty}$-algebras are generalisations of DG (associative) algebras where the higher maps and their constraints measure to what extent associativity only holds up to homotopy. Similarly, $L_{\infty}$-algebras are generalisations of DG Lie algebras whose higher maps measure how much the Jacobi identity is violated. 

To give the precise definition let us consider a graded vector space~$V$ and denote by ${S_{V}}$ the space ${T_{V}}$ divided by the ideal generated by elements of the form $u\otimes v - (-1)^{|u| \, |v|} v\otimes u$ for homogeneous $u,v\in V$. If we write $v_{1}\wedge\ldots\wedge v_{n}$ for the element represented by $v_{1}\otimes\ldots\otimes v_{n}$ in the quotient space ${S_{V}}$, then the coproduct on ${S_{V}}$ is given by 
$$
v_{1}\wedge\ldots\wedge v_{n} \longmapsto \sum_{j=0}^n \sum_{\sigma\in\text{Sh}(j,n)} \varepsilon_{\sigma;v_{1},\ldots,v_{n}} (v_{\sigma(1)}\wedge\ldots\wedge v_{\sigma(j)}) \otimes (v_{\sigma(j+1)}\wedge\ldots\wedge v_{\sigma(n)})
$$
where $\text{Sh}(j,n)$ is the set of permutations~$\sigma$ of~$n$ elements that satisfy $\sigma(1)<\ldots<\sigma(j)$ and $\sigma(j+1)<\ldots<\sigma(n)$, and the sign $\varepsilon_{\sigma;v_{1},\ldots,v_{n}}$ is defined via $v_{\sigma(1)}\wedge\ldots\wedge v_{\sigma(n)} = \varepsilon_{\sigma;v_{1},\ldots,v_{n}} v_{1}\wedge\ldots\wedge v_{n}$. 

With this preparation a \textsl{curved $L_{\infty}$-algebra} is a graded vector space~$V$ together with a codifferential~$\mathfrak{d}$ of degree $+1$ on ${S_{V}}$. Again, the fact that~$\mathfrak{d}$ is a coderivation allows us to equivalently consider a family of maps $\ell_{n}: V[1]^{\wedge n}\rightarrow V[1]$ with constraints coming from the condition $\mathfrak{d}^2=0$. For more details we refer e.\,g.~to~\cite{lm9406095,KSbook}. 

In the present paper we will only be concerned with the special case of a \textsl{DG Lie algebra} where the maps $\ell_{0}$ and $\ell_{n}$ for $n\geqslant 3$ all vanish. If we write
\be\label{dandbracket}
d=\ell_{1} \, , \quad  [u,v] = (-1)^{\widetilde u} \ell_{2}(u\wedge v)
\ee
then the defining conditions read
\begin{align*}
& [u,v] = (-1)^{|u| \, |v|} [v,u] \, , \quad d([u,v]) = [d(u),v] + (-1)^{|u|} [u,d(v)] \, ,\\
& (-1)^{|u| \, |w|} [u,[v,w]] + (-1)^{|u| \, |v|} [v,[w,u]] + (-1)^{|v| \, |w|} [w,[u,v]] = 0 \, .
\end{align*}

\medskip

Given two curved $A_{\infty}$-algebras $(A,\del)$ and $(A',\del')$, a \textsl{(weak) $A_{\infty}$-morphism} between them is a morphism $F\in\Hom({T_{A}},{T_{A'}})$ of degree 0 between the associated codifferential coalgebras, i.\,e.
\be\label{Ainfmorph}
\Delta \circ F = (F\otimes F)\circ \Delta \, , \quad F \circ \del = \del' \circ F \, .
\ee
If we decompose
$$
F = \sum_{m,n\geqslant 0} F_{m}^n \quad\text{where}\quad F_{0}^0 = 1 \, , \quad F_{m}^n \in \Hom(A[1]^{\otimes m}, A'[1]^{\otimes n})
$$
 then the first equation in~\eqref{Ainfmorph} implies that
 $$
 F_{m}^n = \sum_{j_{1}+\ldots+j_{n}=m} F_{j_{1}}^1 \otimes \ldots \otimes F_{j_{n}}^1
 $$
 can be expressed in terms of the maps $F_{n} = F_{n}^1$. We call~$F$ a weak \textsl{$A_{\infty}$-isomorphism} if $F_{1}$ is an isomorphism of vector spaces. A weak $A_{\infty}$-morphism~$F$ between strong $A_{\infty}$-algebras $(A,r_{n})$ and $(A',r'_{n})$ is called an \textsl{$A_{\infty}$-quasi-isomorphism} if $F_{1}$ induces a vector space isomorphism between the cohomologies of $r_{1}$ and $r'_{1}$, and it is called a \textsl{(strong) $A_{\infty}$-morphism} if $F_{m}^n=0$ whenever $n>m$. 
 
Analogous definitions hold for \textsl{weak $L_{\infty}$-morphisms} which are morphisms between codifferential coalgebras ${S_{V}}$ and ${S_{V'}}$ (again, see e.\,g.~\cite{lm9406095,KSbook} for details). 
 
 \medskip
 
The general theory of non-curved $A_{\infty}$- and $L_{\infty}$-algebras is developed considerably further than that of the curved case (see, however, \cite{n0702449,p0905.2621}). In particular, the central \textsl{minimal model theorem}~\cite{k0504437,m9809} that we review next has no known counterpart in the curved case. It will suffice to discuss only $A_{\infty}$-algebras. 

The content of the minimal model theorem is that every $A_{\infty}$-algebra $(A,r_{n})$ is related to an $A_{\infty}$-structure $(H=H_{r_{1}}(A), \widetilde r_{n})$ on $r_{1}$-cohomology via an $A_{\infty}$-quasi-isomorphism $F: (H,\widetilde r_{n})\rightarrow (A,r_{n})$. Furthermore, up to $A_{\infty}$-isomorphisms $(H,\widetilde r_{n}$) is the unique minimal $A_{\infty}$-algebra that is $A_{\infty}$-quasi-isomorphic to $(A,r_{n})$. 

To obtain explicit expressions for the maps $\widetilde r_{n}$ and $F_{n}$ we first choose a vector space decomposition
\be\label{AHBL}
A = H\oplus B\oplus L
\ee
where $B=\text{Im}(r_1)$ and $L$ is the complement of $\text{Ker}(r_{1})$. It follows that~$H$ gives a choice of representatives of elements in $H_{r_{1}}(A)$, and in this sense we identify these two spaces. In particular the choice of decomposition amounts to a choice of a homotopy on~$A$, i.\,e.~a map $G:A\rightarrow A$ with $r_1  G + G r_1 = \id_{A[1]} - \pi_H$ and $G^2 = 0$, where $\pi_{H}:A\rightarrow H$ denotes the projection onto~$H$. We also have projections $\pi_B = r_1 \circ G$ and $\pi_L = G \circ r_1$.
With this setup, the $A_{\infty}$-structure on~$H$ and the associated quasi-isomorphism are given by the following proposition.

\begin{proposition}\label{minmodthm}
Let $(A, \partial)$ be a strong $A_{\infty}$-algebra with $r_{1}$-cohomology~$H$. There is a unique coalgebra morphism $F \in \mathrm{Hom}({T_{H}}, {T_{A}})$ and a unique minimal $A_{\infty}$-structure $\widetilde{\partial} \in \mathrm{Coder}({T_{H}})$ that satisfy the equations
\begin{align}
\partial F &= F \widetilde{\partial}\label{Fmorph} \, ,\\
\widetilde{\partial}_1^1 &= 0 \, , \nonumber \\
F_1^1 &= \iota_H: H[1] \hookrightarrow A[1] \label{F11} \, ,\\
F_n^1 &= -G \sum_{k=2}^n \partial_k^1 F_n^k \, .
\label{Fn1}
\end{align}
\end{proposition}

\begin{proof}
First we will show that the condition that $F$ be an $A_{\infty}$-morphism follows from the conditions above, then we show that $\widetilde{\partial}$ is indeed a codifferential. Since $F$ is a coalgebra morphism, \eqref{Fmorph} reduces to
$$
\partial_1^1F_n^1 + \sum_{l = 2}^{n}\partial_l^1F_n^l = \sum_{k = 1}^{n-1}F_k^1 \widetilde{\partial}_n^k
$$
for all $n \geqslant 1$. We rewrite the above set of equations by splitting them into three parts: 
\begin{align}
\pi_H \Big( \sum_{l = 2}^{n}\partial_l^1F_n^l\Big) &= \widetilde{\partial}_n^1 \label{tildepartial}\, ,\\\pi_B \Big( \partial_1^1 F_n^1 + \sum_{l = 2}^{n}\partial_l^1F_n^l\Big) &= 0 \label{piBminimal}\, ,\\
\pi_L \Big( \sum_{l=2}^n\partial_l^1F_n^l - \sum_{k = 1}^{n-1}F_k^1 \widetilde{\partial}_n^k\Big)  &=0 \, .\label{piLminimal}
\end{align}
The first immediate observation is that $\widetilde{\partial}$ is uniquely determined by~\eqref{tildepartial}. 
Equation~\eqref{piBminimal} follows by employing~\eqref{Fn1} and~\eqref{F11},  and to show~\eqref{piLminimal} we compute
\begin{align*}
\pi_L \Big(\sum_{l=2}^n\partial_l^1F_n^l - \sum_{k = 1}^{n-1}F_k^1 \widetilde{\partial}_n^k\Big) &= \pi_L\Big(\sum_{l=2}^n \partial_l^1F_n^l + G\sum_{k=2}^{n-1}\sum_{r=2}^{k}\partial_r^1F_k^r\widetilde{\partial}_n^k\Big)\\
&= \pi_L\Big(\sum_{l=2}^n \partial_l^1F_n^l + G\sum_{k=2}^{n-1}\sum_{r=2}^{k}\partial_r^1\partial_k^r F_n^k\Big) = 0 \, .
\end{align*}
In the first step we used~\eqref{Fn1}, the second step is the induction step which allows us to commute $\widetilde{\partial}$ through $F$, and in the last step we used $\partial^2 = 0$. 

To show that $\widetilde{\partial}^2 = 0$ we note that
$$
\sum_{k =1}^n\widetilde{\partial}_k^1\widetilde{\partial}_n^k = \sum_{k=1}^n\pi_H\Big(\sum_{l = 2}^{k}\partial_l^1F_k^l \widetilde{\partial}_n^k\Big) = \sum_{k=1}^n\pi_H \Big( \sum_{l = 2}^{k}\partial_l^1\partial_k^l F_n^k\Big) = 0 \, .
$$
\end{proof}

 \medskip
 
In this paper we are concerned with the question of how to construct bulk deformations of an open topological string theory that we describe as an $A_{\infty}$-algebra. More generally, a \textsl{deformation} of a (strong) $A_{\infty}$-algebra $(A,\del)$ is an operator $\delta\in\End({T_{A}})$ such that $(A,\del+\delta)$ is a curved $A_{\infty}$-algebra. We see that~$\delta$ must be a coderivation of degree $+1$, and the condition $(\del+\delta)^2=0$ becomes
\be\label{MC12}
[\del,\delta] + \frac{1}{2}\, [\delta,\delta] = 0 \, ,
\ee
where $[\,\cdot\,,\,\cdot\,]$ denotes the graded commutator. This is the \textsl{Maurer-Cartan equation} for the DG Lie algebra $\text{Coder}({T_{A}})$ with differential $[\del,\,\cdot\,]$ and bracket $[\,\cdot\,,\,\cdot\,]$. Thus solving the Maurer-Cartan equation for the special $L_{\infty}$-algebra $(\text{Coder}({T_{A}}), [\del,\,\cdot\,], [\,\cdot\,,\,\cdot\,])$ is equivalent to solving the deformation problem for the $A_{\infty}$-algebra $(A,\del)$.\footnote{Note that we slightly abuse notation here and below as the bracket $[\,\cdot\,,\,\cdot\,]$ and the map $\ell_{2}$ are only equal up to a sign, see equation~\eqref{dandbracket}.} 

For an arbitrary $L_{\infty}$-algebra $(V,\ell_{n})$ its Maurer-Cartan equation reads
$$
\sum_{n\geqslant 1} \frac{1}{n!}\, \ell_{n}(\delta^{\wedge n}) = 0 \, ,
$$
and we denote by $\MC(V,\ell_{n})$ its space of formal power series solutions $\delta\in V_{1}$ modulo the action of the group generated by \textsl{gauge transformations}
$$
\delta \longmapsto \delta + \sum_{n\geqslant 1}\frac{1}{(n-1)!} \, \ell_{n}(\varphi \wedge \delta^{\wedge (n-1)})
$$
for all $\varphi\in V_{0}$. 

In the case of the deformation problem~\eqref{MC12}, gauge transformations lead to isomorphic $A_{\infty}$-structures $(A,\del)\cong(A,\del+\delta)$. Thus 
$$
\MC(\text{Coder}({T_{A}}), [\del,\,\cdot\,], [\,\cdot\,,\,\cdot\,])
$$
is the space we wish to determine, for which the following important result of~\cite{k9709040,m0001007} will prove useful in the next section. 
\begin{proposition}\label{prop:MCsolutions}
Let $F: (V,\ell_{n})\rightarrow (V',\ell'_{n})$ be an $L_{\infty}$-morphism. Then
\be\label{MCmap}
\delta \longmapsto \sum_{n\geqslant 1} \frac{1}{n!} \, F_{n}(\delta^{\wedge n}) 
\ee
maps elements in $\MC(V,\ell_{n})$ to elements in $\MC(V',\ell'_{n})$. Furthermore, if~$F$ is an $L_{\infty}$-quasi-isomorphism, then~\eqref{MCmap} is an isomorphism. 
\end{proposition}
\begin{proof}
Denote by~$\frak{d}$ the coderivation determined by the higher maps $\ell_n$. Consider a weak coalgebra morphism $M \in \mathrm{End}({S_{V}})$ with the properties $M_0^1 = \delta$ and $M_1^1 = \mathrm{id}_{V[1]}$, where~$\delta$ is a solution to the Maurer-Cartan equation for~$\frak{d}$. Then we have $M_0^n = \frac{1}{n!} \delta^{\wedge n}$ and the Maurer-Cartan equation for~$\frak{d}$ can be rewritten as
$$
(\mathfrak{d}M)_0^1 = 0 \, .
$$

Now let~$\frak{d}'$ denote the coderivation corresponding to the products $\ell_n'$. We have
$$
(\mathfrak{d}' FM)_0^1 = (F \mathfrak{d}M)_0^1 = \sum_{k \geqslant 1}F_k^1 (\mathfrak{d}M)_0^k = 0
$$
where the last equation follows from the fact that $(\mathfrak{d}M)_0^k$ is uniquely determined by $(\mathfrak{d}M)_0^1$. We have thus shown that $\delta' = (FM)_0^1$ solves the Maurer-Cartan equation for $\mathfrak{d}'$. In expanded form this reads
$$
\delta' = \sum_{n \geqslant 1}F_n^1 M_0^n = \sum_{n \geqslant 1}\frac{1}{n!}F_n^1(\delta^{\wedge n}) \, .
$$

If~$F$ is a quasi-isomorphism, then it is an isomorphism between the spaces of first order deformations, and as it admits homotopy inverses, this isomorphism extends to all orders. 
\end{proof}

We also observe that solving~\eqref{MC12} up to gauge transformations only to first order is the same as computing the cohomology
$$
\HH^\bullet(A,\del) = H_{[\del,\,\cdot\,]}(\text{Coder}({T_{A}}))
$$
of the \textsl{Hochschild cochain complex} $(\text{Coder}({T_{A}}), [\del,\,\cdot\,])$. In this sense Hochschild cohomology $\HH^\bullet(A,\del)$ classifies deformations of $(A,\del)$. 

\begin{remark}\label{HHremark}
There is an important subtlety in the definition of the Hochschild cochain complex. As we saw in our discussion of equation~\eqref{manytomany}, coderivations of ${T_{A}}$ are isomorphic to collections of multilinear maps. However, one may either consider finitely or infinitely many multilinear maps, so there are actually Hochschild cochain complexes of the \textsl{first kind} and of the \textsl{second kind},\footnote{There is a much deeper origin of the two different kinds of Hochschild complexes, see~\cite{p0905.2621,pp1010.0982}} 
\begin{align*}
\text{Coder}({T_{A}})^{\text{I}} & \cong \prod_{n\geqslant 0} \Hom(A[1]^{\otimes n}, A[1]) \, , \\
\text{Coder}({T_{A}})^{\text{II}} & \cong \bigoplus_{n\geqslant 0} \Hom(A[1]^{\otimes n}, A[1]) \, .
\end{align*}
Both complexes are endowed with the same differential $[\del,\,\cdot\,]$, but they have different invariance properties: Hochschild cohomology of the first kind is invariant under (strong) $A_{\infty}$-quasi-isomorphisms~\cite{l0204062}, while Hochschild cohomology of the second kind is invariant under weak $A_{\infty}$-isomorphisms~\cite{pp1010.0982}. 

As follows from the discussion of the next section it is Hochschild cohomology of the second kind that describes \textsl{all} bulk deformations. On the other hand, Hochschild cohomology of the first kind cannot always capture all bulk fields, in fact for certain branes it actually vanishes. Thus we simply write $\HH^\bullet(A,\del)$ for Hochschild cohomology of the second kind in this paper. 
\end{remark}

\subsection{B-twisted Landau-Ginzburg models}\label{subsec:LGmodels}

\subsubsection{Open/closed topological field theory}

We begin our discussion of B-twisted Landau-Ginzburg models by recalling that they give examples of full open/closed topological field theories in the sense of~\cite{l0010269,ms0609042}. The on-shell bulk sector~\cite{v1991} of such a model with affine target space $X=\C^N$ and potential $W\in R=\C[x_{1},\ldots,x_{N}]$ is given by the Jacobian\footnote{We always assume that~$W$ has an isolated singularity at the origin, i.\,e.~$\Jac(W)$ is finite-dimensional.}
\be\label{JacW}
\Jac(W) = R/(\del_{x_{1}}W,\ldots,\del_{x_{N}}W) \, .
\ee
This is the cohomology of the off-shell DG Lie algebra 
$$
T_{\text{poly}} = \left( \Gamma(X, \bigwedge T^{(1,0)}X), [-W, \,\cdot\,]_{\text{SN}}, [\,\cdot\,,\,\cdot\,]_{\text{SN}} \right)
$$
of polyvector fields, where $[\,\cdot\,, \,\cdot\,]_{\text{SN}}$ is the Schouten-Nijenhuis bracket extending the Lie bracket from vector fields to polyvector fields; we will recall the definition in section~\ref{sec:bulk-deformedLG}. The on-shell space $\Jac(W)$, viewed as the minimal model of $T_{\text{poly}}$, has a trivial $L_{\infty}$-structure. 

For a bulk field $\phi\in\Jac(W)$ its one-point-correlator is
$$
\langle\phi\rangle = \Res \left[\frac{\phi\, \D x_{1}\wedge\ldots\wedge \D x_{N}}{\del_{x_{1}}W\ldots\del_{x_{N}}W}\right] \, , 
$$
and the bulk topological metric
\be\label{bulktopmet}
\Jac(W) \otimes \Jac(W) \longrightarrow \C \, , \quad \phi_{1}\otimes \phi_{2} \longmapsto \langle\phi_{1}\phi_{2}\rangle
\ee
is known to be non-degenerate~\cite{GandH}. 

\medskip

The on-shell boundary sector is described by the category of matrix factorisations $\MFW$~\cite{kl0210,bhls0305,l0312}. We take its objects, interpreted as D-branes, to be odd square supermatrices~$D$ with polynomial entries such that $D^2=W\cdot e$, where~$e$ is the identity matrix of the same size as~$D$. Such matrix factorisations are also the objects of the off-shell DG category $\DGW$ whose morphisms $D\rightarrow D'$ are polynomial supermatrices~$\psi$ such that the compositions $D'\psi$ and $\psi D$ make sense. The differential on the $\Z_{2}$-graded spaces $\Hom_{\DGW}(D,D')$ is the boundary BRST operator
$$
d: \psi \longmapsto D'\psi - (-1)^{|\psi|} \psi D
$$
for homogeneous~$\psi$. The on-shell D-brane category is defined to be the homotopy category
\be\label{MFW}
\MFW = H_{d}^\bullet(\DGW) \, .
\ee

Given an open string operator $\psi\in\End_{\MFW}(D)$, it was argued in~\cite{kl0305,hl0404} that its one-point-correlator is the \textsl{Kapustin-Li trace}
$$
\langle\psi\rangle_{D} = (-1)^{N\choose 2} \Res \left[ \frac{\str(\psi\, \del_{x_{1}}D\ldots \del_{x_{N}}D)}{\del_{x_{1}}W\ldots\del_{x_{N}}W} \right] \, , 
$$
and the boundary topological metric is the \textsl{Kapustin-Li pairing}
\be\label{bdrytopmet}
\Hom_{\MFW}(D,D') \otimes \Hom_{\MFW}(D',D) \longrightarrow \C \, , \quad \psi_{1}\otimes \psi_{2} \longmapsto \langle \psi_{1}\psi_{2} \rangle_{D} \, .
\ee
The non-degeneracy of this pairing was proved in~\cite{m0912.1629,dm1004.0687}. 

\medskip

To complete the structure of open/closed topological field theory for Landau-Ginzburg models we also need the (on-shell) bulk-boundary and boundary-bulk maps~\cite{kr0405232}:
\begin{align}
\label{bubo} & \beta^{(D)}_{\text{bubo}}: \Jac(W) \longrightarrow \End_{\MFW}(D) \, , \quad \phi \longmapsto \phi \cdot e \, ,\\
\label{bobu} & \beta^{(D)}_{\text{bobu}}: \End_{\MFW}(D) \longrightarrow \Jac(W) \, , \quad \psi \longmapsto (-1)^{N\choose 2} \str(\psi\, \del_{x_{1}}D\ldots \del_{x_{N}}D) \, .
\end{align}

It is mostly straightforward to show that the data~\eqref{JacW}--\eqref{bobu} satisfy all the axioms of an open/closed topological field theory of~\cite{l0010269,ms0609042}. Apart from the non-degeneracy of the Kapustin-Li pairing, the only exception is the Cardy condition
$$
\str\big( \psi \longmapsto \psi_{2}\psi\psi_{1} \big) = \big\langle \beta^{(D)}_{\text{bobu}}(\psi_{1}) \, \beta^{(D')}_{\text{bobu}}(\psi_{2}) \big\rangle
$$
for open string operators $\psi_{1}:D\rightarrow D$, $\psi_{2}:D'\rightarrow D'$ and $\psi:D\rightarrow D'$, which was proved in~\cite{pv1002.2116}.

\subsubsection{Open topological string theory}\label{oTSTforLG}

Our main motivation is to enrich the above structure of open/closed topological \textsl{field} theory for Landau-Ginzburg models to that of open/closed topological \textsl{string} theory. In the next section we will face this task by studying bulk-deformed open topological string theory; here we very briefly review the construction of the pure boundary sector. 

As we recalled in the introduction, constructing open topological string theories from Landau-Ginzburg models amounts to endowing the on-shell open string spaces
$$
H=\End_{\MFW}(D)
$$
with the correct $A_{\infty}$-structures for all matrix factorisations~$D$. The first natural step is to use the minimal model construction of proposition~\ref{minmodthm} to transport the DG structure~$\del$ (with $r_{1}=[D,\,\cdot\,]$ and $r_{2}$ coming from matrix multiplication) of the off-shell space 
$$
A=\End_{\DGW}(D)
$$
to an $A_{\infty}$-structure on~$H$. 

Unfortunately, the thus obtained $A_{\infty}$-structure on the on-shell space~$H$ is generically not cyclic with respect to the Kapustin-Li pairing~\eqref{bdrytopmet} and does hence not correctly encode the topological string amplitudes. As explained in detail in~\cite{c0904.0862} one can however use methods of formal non-commutative geometry~\cite{kontsevich1993,ks0606241} to construct the full structure of open topological string theory. The main idea of this approach is to enhance the Kapustin-Li pairing~\eqref{bdrytopmet} to a non-commutative symplectic form that is suitably compatible with the generic $A_{\infty}$-structure on~$H$ and then use a version of the Darboux theorem to transform it  to a cyclic structure $\widetilde\del$. An alternative, non-explicit existence proof of cyclic $A_{\infty}$-structures on~$H$ was sketched in~\cite{dm1004.0687}, and in subsection~\ref{omegaconstruction} below we will sketch how to drastically improve the above-mentioned explicit construction using symplectic forms. 

\medskip

The goal of our study in the present paper are bulk-induced deformations of the open topological string theory encoded in the on-shell $A_{\infty}$-structure $(H,\widetilde\del)$. We have recalled in this section how this originates from the off-shell DG structure $(A,\del)$, and also that deformations of $A_{\infty}$-structures are classified by Hochschild cohomology. Thus it is a welcome fact that the latter is given by the space of bulk fields:
$$
\HH^\bullet(A,\del) \cong \Jac(W) \, .
$$
This result will be proved along with theorem~\ref{Th1} in the next section. We further note that the Hochschild cohomology of the full category $\DGW$ is also given by $\Jac(W)$ as shown in~\cite{d0904.4713}.

\section{Bulk-deformed Landau-Ginzburg models}\label{sec:bulk-deformedLG}

The notion of bulk-induced deformations of the on-shell boundary space $(H, \widetilde{\partial})$ relies on the existence of an $L_{\infty}$-morphism 
\be\label{Lmorph} 
L: (T_{\text{poly}}, [-W, \,\cdot\,]_{\text{SN}}, [\,\cdot\,,\,\cdot\,]_{\text{SN}}) \longrightarrow (\mathrm{Coder}({T_{H}}), [\widetilde{\partial}, \,\cdot\,], [\,\cdot\,,\,\cdot\,])
\ee
from the DG Lie algebra of off-shell bulk fields to the DG Lie algebra of coderivations on the boundary side. A \textsl{bulk-induced deformation} is then defined as the image under $L$ of a deformation of the pure bulk theory. In this section we give an explicit construction of $L$. 

The map~\eqref{Lmorph} splits naturally as the composition of two $L_{\infty}$-morphisms that we discuss in subsections~\ref{WDQ} and~\ref{transbulkonshell}, respectively. The first is an $L_{\infty}$-quasi-isomorphism 
\be\label{L1}
(T_{\text{poly}}, [-W, \,\cdot\,]_{\text{SN}}, [\,\cdot\,,\,\cdot\,]_{\text{SN}}) \longrightarrow (\mathrm{Coder}({T_{A}}), [\partial, \,\cdot\,], [\,\cdot\,,\,\cdot\,])
\ee
and thus identifies the two deformation problems. It can be viewed as a ``weak'' version of Kontsevich's construction for (local) deformation quantisation, or rather 
 its complex cousin. 
The second $L_{\infty}$-morphism 
\be\label{L2}
(\mathrm{Coder}({T_{A}}), [\partial, \,\cdot\,], [\,\cdot\,,\,\cdot\,]) \longrightarrow (\mathrm{Coder}({T_{H}}), [\widetilde\partial, \,\cdot\,], [\,\cdot\,,\,\cdot\,])
\ee
transports off-shell deformations on-shell and can be viewed as the $L_{\infty}$-formulation and enhancement of the homological perturbation lemma.\footnote{As we are concerned with the Hochschild cochain complex of the second kind, the map~\eqref{L1} is a quasi-isomorphism, while the map~\eqref{L2} is generically not. However, if we used the Hochschild cochain complex of the first kind, \eqref{L2} would always be a quasi-isomorphism, but generically not~\eqref{L1}. See also remark~\ref{HHremark}.}

\subsection{Weak deformation quantisation}\label{WDQ}

We continue to use the notation of the previous section; in particular we fix a potential $W\in R=\C[x_{1},\ldots,x_{N}]$, a matrix factorisation~$D$, and with it the off-shell and on-shell open string algebras $A=\End_{\DGW}(D)$ and $H=\End_{\MFW}(D)$, respectively. 

As follows from our earlier discussion, bulk deformations of B-twisted Landau-Ginzburg models are solutions $\gamma \in T_{\text{poly}}$ of the Maurer-Cartan equation
\begin{equation}\label{MCTpoly}
[-W, \gamma]_{\text{SN}} + \frac{1}{2}[\gamma, \gamma]_{\text{SN}} = 0 \, ,
\end{equation}
where the Schouten-Nijenhuis bracket on $T_{\text{poly}}$ is given by
\begin{align*}
& [\zeta_{1}\wedge\ldots\wedge \zeta_{m}, \xi_{1}\wedge\ldots\wedge \xi_{n} ]_{\text{SN}} = 
\sum_{i=1}^m \sum_{j=1}^n (-1)^{i+j} [\zeta_{i},\xi_{j}] \\
&
\qquad \cdot \zeta_{1}\wedge\ldots\wedge\zeta_{i-1}\wedge\zeta_{i+1}\wedge\ldots\wedge \zeta_{m}
\wedge \xi_{1}\wedge\ldots\wedge\xi_{j-1}\wedge\xi_{j+1}\wedge\ldots\wedge \xi_{n}
\end{align*}
and
$$
[\zeta_{1}\wedge\ldots\wedge \zeta_{m}, f ]_{\text{SN}} = \sum_{i=1}^m (-1)^i \zeta_{i}(f) \zeta_{1}\wedge\ldots\wedge\zeta_{i-1}\wedge\zeta_{i+1}\wedge\ldots\wedge \zeta_{m}
$$
for $\zeta_{i}, \xi_{j}\in \Gamma(X,T^{(1,0)}X)$ and $f\in\Gamma(X,\mathcal O_{X})$. 

We restrict our attention to formal power series in a set of parameters~$t$, i.\,e.~$\gamma = \sum_{i \geqslant 1} t^i \gamma^{(i)}$. This assumption allows to solve (\ref{MCTpoly}) perturbatively. At first order the equation reads
\begin{equation}
[-W, \gamma^{(1)}]_{\text{SN}} = 0 \, ,
\label{MCTpoly1}
\end{equation}
and hence we have, up to gauge transformations, $\gamma^{(1)} \in \Jac(W)$, the on-shell bulk space.  One simplicity of affine Landau-Ginzburg models lies in the fact that the solutions of (\ref{MCTpoly1}) are automatically solutions of the full Maurer-Cartan equation, as the Schouten-Nijenhuis bracket of two functions vanishes. Having thus fully solved the deformation problem in the bulk sector, the problem of computing bulk-induced deformations reduces to that of transporting bulk deformations appropriately to the boundary sector. The first main result in accomplishing this task is the following, which at its core is a weak version of deformation quantisation. 
\begin{theorem}\label{Th1}
$(\mathrm{Coder}({T_{A}}), [\partial, \,\cdot\,], [\,\cdot\,,\,\cdot\,])$ and $(T_{\text{poly}}, [-W,\,\cdot\,]_{\text{SN}}, [\,\cdot\,,\,\cdot\,]_{\text{SN}})$ are quasi-isomorphic as $L_{\infty}$-algebras.
\end{theorem}
To understand this result a crucial role is played by the curved associative algebra
$$
(A, \partial_0 + \partial_2) \quad \text{with}\quad \partial_0^1 = -W \cdot e \, ,
$$
where~$e$ denotes the identity matrix of the same size as~$D$. In the following we will write $\partial$ as $\partial_1 + \partial_2$ to distinguish it from $\partial_0 + \partial_2$. We proceed to give a constructive proof of the above theorem. For the purpose of clarity we shall split it into the following three parts. 

\begin{lemma}\label{Tmaplemma}
There is an $L_{\infty}$-quasi-isomorphism
$$
(\mathrm{Coder}({T_{A}}), [\partial_1 + \partial_2,\,\cdot\,], [\,\cdot\,,\,\cdot\,]) \longrightarrow (\mathrm{Coder}({T_{A}}), [\partial_0 + \partial_2,\,\cdot\,], [\,\cdot\,,\,\cdot\,]) \, . 
$$
\end{lemma}

\begin{proof} 
We define a weak coalgebra isomorphism~$T$ (for off-shell ``tadpole-cancellation'', i.\,e.~$T$ is a weak $A_\infty$-isomorphism to a non-curved $A_\infty$-algebra) via its fundamental maps $T_{0}=-D$ and $T_{1}=\id_{A[1]}$, and we compute
\begin{align*}
[(\partial_1 + \partial_2)\circ T]_0^1 &= [D, -D] + D(-D) = -W\cdot e = [T\circ(\partial_0 + \partial_2)]_0^1\, ,\\
[(\partial_1 + \partial_2)\circ T]_1^1&= [D, \,\cdot\,] - \partial_2^1(D \otimes \id_{A[1]} + \id_{A[1]} \otimes D) = 0 = [T\circ(\partial_0 + \partial_2)]_1^1 \, ,\\
[(\partial_1 + \partial_2)\circ T]_2^1 &= \partial_2^1 =  [T\circ(\partial_0 + \partial_2)]_2^1 \, .
\end{align*}
Hence we have $(\partial_1 + \partial_2)\circ T = T\circ(\partial_0 + \partial_2)$, i.\,e.~$T$ is a weak $A_{\infty}$-isomorphism from $(A, \partial_0 + \partial_2)$ to $(A, \partial_1 + \partial_2)$. The $L_{\infty}$-morphism is then given by the adjoint action of~$T$. That it is a quasi-isomorphism follows from the fact that Hochschild cohomology of the second kind is invariant under curved DG isomorphisms~\cite{pp1010.0982}.
\end{proof}

The next step in the proof of theorem~\ref{Th1} is an $L_{\infty}$-version of Morita equivalence:
\begin{lemma}\label{Morita}
There is an $L_{\infty}$-quasi-isomorphism
$$
(\mathrm{Coder}({T_{A}}), [\partial_0 + \partial_2,\,\cdot\,], [\,\cdot\,,\,\cdot\,]) \longrightarrow 
(\mathrm{Coder}({T_{R}}), [\widehat{\partial}_0 + \widehat{\partial}_2,\,\cdot\,], [\,\cdot\,,\,\cdot\,]) \, ,
$$
where $\widehat{\partial}_0 + \widehat{\partial}_2$ is the codifferential on ${T_{R}}$ induced from $\del_{0}+\del_{2}$.
\end{lemma}

\begin{proof}
First we construct the cotrace map $C_1^1 = \mathrm{cotr} : \mathrm{Coder}({T_{R}}) \rightarrow \mathrm{Coder}({T_{A}})$, then we show that the coalgebra morphism~$C$ defined by $C^1 = C_1^1$ is the desired $L_{\infty}$-quasi-isomorphism. The cotrace map is a slightly modified version of the case for ungraded algebras, see e.\,g.~\cite{Loday}: for $\Phi \in \mathrm{Coder}({T_{R}})$ we define $\mathrm{cotr}(\Phi) \in \mathrm{Coder}({T_{A}})$ via $\mathrm{cotr}(\Phi)_0^1 = \Phi_0^1 \cdot e$ and
$$
\K{(}{)}{\mathrm{cotr}(\Phi)_m^1(a_1\otimes\ldots\otimes a_m)}_{kl} = \sum_{i_1,\ldots,i_{m-1}=1}^{2n}\Phi_m^1((\sigma^{m+1} a_1)_{ki_1}\otimes\ldots\otimes (\sigma a_m)_{i_{m-1}l})
$$
where $\sigma$ is the unique matrix that for homogenous elements $a \in A$ satisfies $\sigma a = (-1)^{|a|}a \sigma$, and $2n$ is the size of~$D$. It is then straightforward to show that
\begin{equation}\label{CotrLie}
[\mathrm{cotr}(\Phi_1), \mathrm{cotr}(\Phi_2)] = \mathrm{cotr}([\Phi_1,\Phi_2]) \, ,
\end{equation}
i.\,e.~$\mathrm{cotr}$ is a map of Lie algebras. In order to show that~$C$ is an $L_{\infty}$-morphism it then suffices to check that $\mathrm{cotr}$ is a map of complexes. This however follows immediately from \eqref{CotrLie} by noting that
$$
\partial_0 + \partial_2 = \mathrm{cotr}(\widehat{\partial}_0 + \widehat{\partial}_2) \, .
$$
That $C_1^1$ is indeed a quasi-isomorphism  follows from a spectral sequence argument and will be shown together with the next proposition.
\end{proof}

Now we arrive at the last and hardest step in the proof of theorem~\ref{Th1}.

\begin{proposition}\label{wDQ}
$(\mathrm{Coder}({T_{R}}), [\widehat{\partial}_0 + \widehat{\partial}_2,\,\cdot\,], [\,\cdot\,,\,\cdot\,])$ and $(T_{\text{poly}}, [-W,\,\cdot\,]_{\text{SN}}, [\,\cdot\,,\,\cdot\,]_{\text{SN}})$ are quasi-isomorphic as $L_{\infty}$-algebras.
\end{proposition}

Fortunately we can build on Kontsevich's result on deformation quantisation~\cite{k9709040} which says that the above is true for the case $W = 0$. More precisely, Kontsevich's result concerns polyvector fields on $\R^d$, but as we are dealing with affine space, his result extends trivially. 

Before delving into the proof of proposition~\ref{wDQ}, let us briefly recall the aim and method of deformation quantisation. One is interested in quantising a classical system described by a phase space $M$ (a real smooth manifold, for simplicity consider $M = \R^d$) or alternatively rather by its commutative, associative algebra of observables $(C^{\infty}(M, \R), \,\cdot\,)$. \textsl{Quantising} in this context amounts to deforming the multiplication ``$\,\cdot\,$'' to an associative, but not necessarily commutative product ``$\star$'', in order to pass to the algebra of quantum observables $(C^{\infty}(M,\R)[\![\hbar]\!], \star)$ while postponing the task of representing it on a Hilbert space. Below we will drop formal parameters like~$\hbar$ from our notation.

To rephrase the problem in a more compact notation, we denote ``$\,\cdot\,$'' by $\partial_2^1$ so that the deformation problem becomes that of solving the Maurer-Cartan equation of the DG Lie algebra 
$(\mathrm{Coder}(T_{C^{\infty}(M, \R)}), [\partial_2,\,\cdot\,], [\,\cdot\,,\,\cdot\,])$. Kontsevich's solution is to construct an $L_{\infty}$-quasi-isomorphism
$$
K : (\Gamma(M, \bigwedge TM), [\,\cdot\,,\,\cdot\,]_{\text{SN}}) \longrightarrow (\mathrm{Coder}(T_{C^{\infty}(M, \R)}), [\partial_2,\,\cdot\,], [\,\cdot\,,\,\cdot\,]) \, .
$$
Thus by proposition~\ref{prop:MCsolutions} every perturbatively deformed product ``$\star$'' originates from a Poisson structure on~$M$, i.\,e.~a degree $2$ polyvector field $\alpha$ which satisfies the Maurer-Cartan equation $[\alpha, \alpha]_{\text{SN}} = 0$. 

As already observed, we can view our theorem~\ref{Th1} as a generalisation of deformation quantisation: endowed with a non-trivial differential $[-W,\,\cdot\,]$, the DG Lie algebra of polyvector fields now governs deformations of a DG algebra, and not just a commutative associative one. 

The proof of proposition~\ref{wDQ} splits into two parts. First we show that the $L_{\infty}$-morphism~$K$ can be extended to the case $W \neq 0$, then we show that it is still a quasi-isomorphism.

\begin{proof}
We start by recalling Kontsevich's construction. The $L_{\infty}$-morphism~$K$ is given by
$$
(K_n^1(\gamma_1\wedge \ldots \wedge \gamma_n))_m^1 = \sum_{\Gamma \in \mathcal{G}(n,m)} w_{\Gamma} \, U_{\Gamma}
$$
where $\mathcal{G}(n,m)$ denotes the set of certain directed graphs $\Gamma$ to which in turn we will associate certain weights $w_{\Gamma}\in\R$ and multilinear maps $U_{\Gamma}$ on $R[1]$. To describe these, consider the unit disc $D$ in the complex plane. Choose $m$ marked points $q_{\bar 1},\ldots, q_{\bar m}$ (which we associate with functions $f_{\bar 1},\ldots,f_{\bar m}$) on the boundary $\partial D$ and $n$ marked points $p_1,\ldots,p_n$ (which we associate with polyvector fields $\gamma_{1},\ldots,\gamma_{n}$) in the interior. These $m+n$ marked points coincide with the vertices of the graph $\Gamma\in\mathcal{G}(n,m)$. 

The possible edges between vertices are constrained by the following rules: (i) for every polyvector field $\gamma_k$, there are precisely $\widetilde{\gamma}_k$ edges $e_k^1,\ldots,e^{\widetilde{\gamma}_k}_k$ starting at $p_k$ and ending on distinct marked points different from $p_k$, (ii) each marked point on the boundary has zero outgoing edges and at least one incoming edge, (iii) the total number of edges is $\mathrm{dim}(C^{n,m}) = 2n + m - 2 \geqslant 0$, where we denote by $C^{n,m}$ the moduli space of the above described marked points on the unit disc with a choice of orientation. Here we run clockwise around the circle and the orientation is well-defined by omitting the point $\I \in \partial D$. This special point is to be viewed as representing the ``out-state''.

To construct the map $U_{\Gamma}\in\Hom(R[1]^{\otimes m},R[1])$ for fixed polyvector fields $\gamma_{1},\ldots,\gamma_{n}$, one views the edges ending on a vertex as the action of the coordinate vector fields on the function associated to the vertex and then takes the product over all such actions. More precisely, if we write
$$
\gamma_{i} = \gamma^{j_{i,1}\ldots j_{i,\widetilde\gamma_{i}}} \del_{j_{i,1}} \wedge\ldots\wedge \del_{j_{i,\widetilde\gamma_{i}}}
$$
and denote by $\Gamma_{\bullet\rightarrow k}$ the set of edges ending on vertex~$k$, then we have
$$
U_{\Gamma}(f_{1}\otimes\ldots\otimes f_{m}) = 
\sum_{I} \Big[ \prod_{i=1}^n \Big( \prod_{e\in\Gamma_{\bullet\rightarrow i}} \del_{I(e)} \Big) \gamma_{i}^{I(e_{i}^1)\ldots I(e_{i}^{\widetilde\gamma_{i}})} \Big]
\Big[ \prod_{\bar \jmath=\bar 1}^{\bar m} \Big( \prod_{e\in\Gamma_{\bullet\rightarrow \bar\jmath}} \del_{I(e)} \Big) f_{\bar \jmath} \Big]
$$
where the sum is over all maps $I:\Gamma_{1}\rightarrow\{1,\ldots,d\}$. 

The weights $w_{\Gamma}$ are certain integrals over the moduli space $C^{n,m}$. In order to understand these, consider for every edge $e^r_k$, the angle map $\varphi_{e^r_k}:  \overline{D}\times \overline{D} \rightarrow (0, 2\pi]$ measuring the (clockwise) angle between the edge $e_k^r$ and the line connecting $p_k$ to the out-state at~$\I$.\footnote{More precisely the angle is measured with respect to the hyperbolic metric, and the edges are the associated geodesics. This however does not influence our discussion.} If we denote by $\iota : C^{n,m} \rightarrow D^{\times n} \times \partial D^{\times m}$ the canonical embedding of the moduli space, then the weights are given by
\begin{equation}\label{weight}
w_{\Gamma} = \frac{1}{(2\pi)^{2n + m-2}}\int_{\iota(\overline{C}^{n,m})} \bigwedge_{k=1}^n(\D\varphi_{e_1^k}\wedge \ldots \wedge \D\varphi_{e_{\widetilde{\gamma}_k}^k}) \, .
\end{equation}

We are now in a position to start with the proof proper. Let $\widehat{\frak{d}}_1^1 = [\widehat{\partial}_0 + \widehat{\partial}_2,\,\cdot\,]$ and define $\widehat{\frak{d}}_2^1$ via $\widehat{\frak{d}}_2^1(\Phi_1 \wedge \Phi_2) = (-1)^{\widetilde{\Phi}_1}[\Phi_1, \Phi_2]$. We want to show that~$K$ continues to be an $L_{\infty}$-quasi-isomorphism also in the curved case, i.\,e. 
$$
K_n^1l_n^n + K_{n-1}^1l_n^{n-1} = \widehat{\frak{d}}_1^1 K_n^{1} + \widehat{\frak{d}}_{2}^1K_n^2 \, ,
$$
where we denote the DG Lie algebra structure on $T_{\text{poly}}$ by the maps $l_{m}^n$. If we define the coderivation $c$ given by $c^1 = c_1^1 = \{\partial_0, \,\cdot\,\}$, then by Kontsevich's result  the above reduces to
$$
Kl_1 = cK
$$
which in expanded form reads
\begin{align}
&\sum_{k=0}^{n}(-1)^{\sum_{s=1}^{k-1} \widetilde{\gamma}_s}(K_n^1(\gamma_1 \wedge \ldots \wedge [-W, \gamma_k]\wedge\ldots\wedge \gamma_n))_m^1(f_1 \otimes \ldots \otimes f_m) \nonumber \\
&= \sum_{l=0}^m (-1)^{l}(K_n^1(\gamma_1 \wedge\ldots\wedge \gamma_n))_{m+1}^1(f_1 \otimes \ldots \otimes f_l \otimes (-W)\otimes \ldots \otimes f_m) \, . \label{rhs-weak}
\end{align}

\begin{figure}[t]
$$
\begin{tikzpicture}[scale=1.2,baseline=-0.1cm, inner sep=1mm, >=stealth]
\draw[very thick, dotted] (0,0) circle (4);  

\node (NP) at (90:4) [circle,draw=black,fill=white,label=above:\footnotesize{i}] {};

\node (f1) at (72:4) [circle,draw=black,fill=black,label=above:$f_{1}$] {};
\node (f2) at (-40:4) [circle,draw=black,fill=black,label=right:$f_{2}$] {};
\node (f3) at (-125:4) [circle,draw=black,fill=black,label=below:$f_{3}$] {};
\node (f4) at (145:4) [circle,draw=black,fill=black,label=left:$f_{4}$] {};
\node (W) at (20:4) [circle,draw=black,fill=black,label=right:$-W$] {};

\node (n1) at (1.5,2.3) [circle,draw=black,fill= white] {$\footnotesize \gamma_{1}$};
\node (n2) at (2.1,-0.9) [circle,draw=black,fill= white] {$\footnotesize \gamma_{2}$};
\node (n3) at (1.1,-3.0) [circle,draw=black,fill= white] {$\footnotesize \gamma_{3}$};
\node (n4) at (-2.7,-0.9) [circle,draw=black,fill= white] {$\footnotesize \gamma_{4}$};
\node (n5) at (-1.9,1.9) [circle,draw=black,fill= white] {$\footnotesize \gamma_{5}$};

\draw (2.17,-0.35) node {{\footnotesize $\varphi_{e_{2}^1}$}};
\draw[<-] (2.55,-0.2) arc (40:130:5mm); 

\draw[dashed] (n2) to (NP); 

\node (e1) at (-0.3,-1.3) [circle,draw=white,fill= white] {$e_{2}^4$}; 
\node (e2) at (-0.3,-0.5) [circle,draw=white,fill= white] {$e_{4}^2$}; 

\draw[->, very thick] (n1) -- (f1) node[midway,left=0] {$e_{1}^1$}; 

\draw[->, very thick] (n2) -- (W) node[midway,left=1.9] {$e_{2}^1$}; 
\draw[->, very thick] (n2) -- (f2) node[midway,right=0] {$e_{2}^2$}; 
\draw[->, very thick] (n2) -- (n3) node[midway,right=0] {$e_{2}^3$}; 
\draw[->, very thick] (n2) to [out=-175,in=-5] (n4); 

\draw[->, very thick] (n3) -- (f2) node[midway,above] {$e_{3}^1$}; 
\draw[->, very thick] (n3) -- (n4) node[midway,below] {$e_{3}^2$}; 

\draw[->, very thick] (n4) -- (n1) node[midway,above] {$e_{4}^1$}; 
\draw[->, very thick] (n4) to [out=5,in=175] (n2); 
\draw[->, very thick] (n4) -- (f3) node[midway,right] {$e_{4}^3$}; 
\draw[->, very thick] (n4) -- (n5) node[midway,left] {$e_{4}^4$}; 

\draw[->, very thick] (n5) -- (f1) node[midway,above] {$e_{5}^1$}; 
\draw[->, very thick] (n5) -- (f4) node[midway,below] {$e_{5}^2$}; 

\end{tikzpicture} 
$$
\caption{A graph contributing to $K_{5}^1(\gamma_{1}\wedge\ldots\wedge\gamma_{5})_{5}^1(f_{1}\otimes(-W)\otimes f_{2}\otimes f_{3}\otimes f_{4})$.} 
\label{diskpic} 
\end{figure}
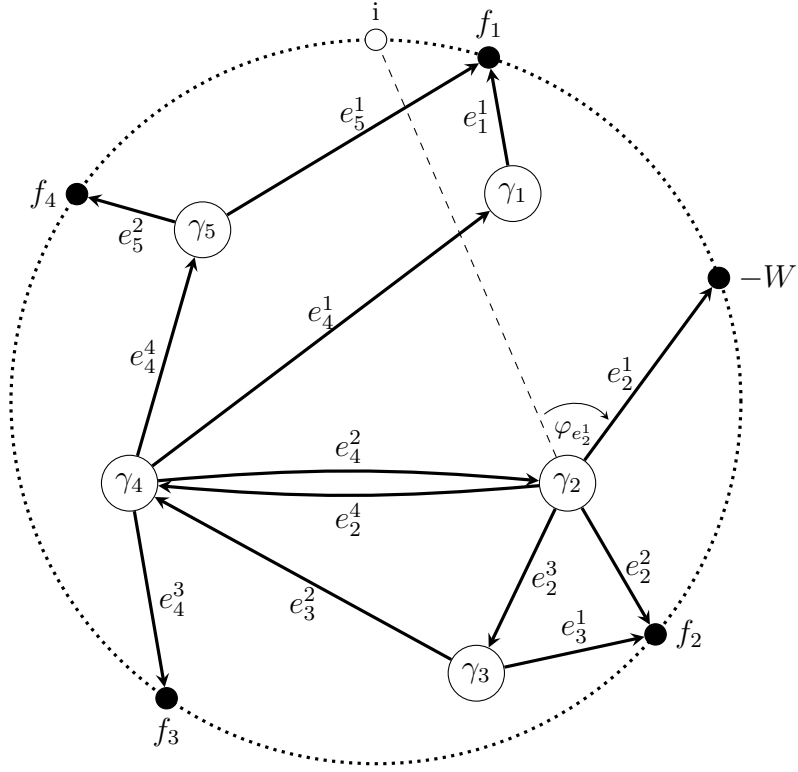

We will now analyse the right-hand side to find that it is the same as the left-hand side; see figure~\ref{diskpic} for a visualisation. Fix a graph $\Gamma \in \mathcal{G}(n, m+1)$ and consider the first summand on the right-hand side of (\ref{rhs-weak}). Pick an edge $e_k^r$ ending on $-W$ and carry the corresponding differential form $\D\phi_{e_k^{r}}$ to the very left of the integral~\eqref{weight}. This results in picking up a sign $\mu=(-1)^{\sum_{s = 1}^{k-1}\widetilde{\gamma_s}}(-1)^{r-1}$. Now consider the $l$-th term in (\ref{rhs-weak}). This differs from the first term by a sign $(-1)^l$ which corresponds to the determinant of the Jacobian of the map transforming $(q_1,\ldots, q_{m+1}) \mapsto (q_2, \ldots, q_l, q_1, \ldots, q_{m+1})$ contributing to the weight. This sign cancels the sign present in the sum, therefore performing the sum over $l$ yields an integral of the angle $\varphi_{e_k^r}$ over $(0,2\pi]$ which decouples from the rest and yields $2\pi$. This is then absorbed by a $2\pi$ in the denominator of the prefactor of (\ref{weight}). We are then left with an integral over $\overline{C}^{n,m}$.  The sign~$\mu$ is the product of the sign present on the left-hand side of the equation times the sign coming from the Schouten-Nijenhuis bracket, and we see that~\eqref{rhs-weak} indeed holds true. 
 
\medskip
 
We will now prove that $K_1^1$ is a quasi-isomorphism. 
Consider $\gamma \in T_{\text{poly}}$ of degree $\widetilde{\gamma} = n$. By construction $(K_1^1(\gamma))^1(f_1 \otimes \ldots \otimes f_m)$ is non-zero only for $m = n$ and is given by
$$
(K_1^1(\gamma))_n^1(f_1 \otimes \ldots \otimes f_n) = \frac{1}{n!}\sum_{\sigma \in S_n} \text{sign}(\sigma)\gamma^{i_1\ldots i_n}\prod_{k = 1}^n \partial_{i_{\sigma(k)}} f_k \, .
$$
By the Hochschild-Kostant-Rosenberg theorem $K_1^1$ is a quasi-isomorphism from $(T_{\text{poly}}, 0)$ to $(\mathrm{Coder}({T_{R}}), [\widehat{\partial}_2, \,\cdot\,])$. 

To show that $K_{1}^1$ is also a quasi-isomorphism in our case $W\neq 0$, the strategy is to view $(\mathrm{Coder}({T_{R}}), [\widehat{\partial}_0 + \widehat{\partial}_2, \,\cdot\,])$ as a bicomplex after choosing appropriate linear combinations of tensor and tilde degrees. We then choose to compute the associated spectral sequence whose first page is $[\widehat{\partial}_2, \,\cdot\,]$-cohomology. The spectral sequence degenerates at the second page yielding
$$
H_{ [\widehat{\partial}_0 + \widehat{\partial}_2, \,\cdot\,]}(\mathrm{Coder}({T_{R}})) = \HH^{\bullet}(R, \widehat{\partial}_0 + \widehat{\partial}_2) = \Jac(W) \, ,
$$
see appendix~\ref{appss} for details. As explained also in~\cite{pp1010.0982,ct1007.2679}, the chosen spectral sequence computes Hochschild cohomology of the second kind. This concludes the proof of proposition~\ref{wDQ}. 

Essentially the same argument is used if we replace $\mathrm{Coder}({T_{R}})$ with $\mathrm{Coder}({T_{A}})$ in the setting of lemma~\ref{Morita}. The first page then is classical Morita equivalence whose proof is the same in the $\mathbb{Z}_2$-graded case. 
\end{proof}

\subsubsection{Constructing the off-shell Kapustin-Li pairing}\label{omegaconstruction}

We close this subsection with a short discussion of another application in which the tadpole-cancellation map~$T$ introduced in the proof of lemma~\ref{Tmaplemma} plays a prominent role. This application concerns the construction of the off-shell or $A_{\infty}$-enrichment of the Kapustin-Li trace and pairing, and it is independent from the remainder of the present paper. To be brief we assume familiarity with basic non-commutative geometry~\cite{Loday, ks0606241} and its relation to $A_{\infty}$-theory; we will use the notation of~\cite{l0507222,c0904.0862}. 

As explained in~\cite{c0904.0862}, one way to construct the Calabi-Yau $A_{\infty}$-structure encoding the open topological string theory on the on-shell space~$H$ is to first find a non-commutative homologically symplectic form~$\omega$ that is the $A_{\infty}$-version of the Kapustin-Li pairing and satisfies the generalised cyclicity condition $L_{Q}\omega=0$, where~$Q$ encodes the DG structure on~$A$. Then in a second step one applies the Darboux theorem and thus pushes~$Q$ forward to the correct cyclic $A_{\infty}$-structure. In the approach of~\cite{c0904.0862} the first step is the computationally much more challenging one, and $\omega$ was constructed only perturbatively by an algorithm that is applied case by case. Now we will describe how to obtain an explicit and general expression for~$\omega$. 

It turns out that on the off-shell or $A_{\infty}$-level there is an interesting subtlety in the relation between the Kapustin-Li trace and pairing that partly arises from the relation between Hochschild and cyclic homology. We recall that the Hochschild chain complex $(\text{C}_{\bullet}(A),b)$ of an $A_{\infty}$-algebra $(A,\del)$ is dual to $(\mathcal C^1(B_{A}), L_{Q})$, where $Q=\del^\vee$ and we agree to indicate the non-vanishing $A_{\infty}$-products as indices in~$Q$ and~$b$. Then it is straightforward to construct the following sequence of maps and show that they are all quasi-isomorphisms: 
\begin{align}\label{AKLtr}
(\mathcal C^1(B_{A}), L_{Q_{1,2}})^* & \longrightarrow (\text{C}_{\bullet}(A),b_{1,2}) \stackrel{T_{*}} {\longrightarrow} (\text{C}_{\bullet}(A),b_{0,2}) \stackrel{\str} {\longrightarrow} (\text{C}_{\bullet}(R),\widehat b_{0,2}) \nonumber \\
 & \stackrel{\text{HKR}} {\longrightarrow}(\Omega^\bullet(\C^N), \D W\wedge(\,\cdot\,))  \stackrel{\text{Res}} {\longrightarrow} \C \, .
\end{align}
In this way we obtain an explicit expression for a 1-form $\theta\in \mathcal C^1(B_{A})$ whose constant part agrees with the Kapustin-Li trace, and by construction it satisfies $L_{Q_{1,2}}\theta=0$. The last four maps in~\eqref{AKLtr} were also independently constructed in~\cite{s0904.1339}. 

To arrive at an expression for the $A_{\infty}$-version of the Kapustin-Li pairing, a natural guess is to find a solution~$\omega'$ of the equation $L_{X\cup Y}\theta=\omega'(X,Y)$ for all non-commutative vector fields $X,Y$. However, one finds that while $L_{Q_{1,2}}\omega'=0$ holds by construction, such an~$\omega'$ is not necessarily homologically symplectic. (Indeed, $\omega'$ is precisely of the form of the special solution of $L_{Q_{1,2}}\omega'=0$ discussed at the end of section~2 in~\cite{c0904.0862}.) 

In order to construct the true off-shell version~$\omega$ of the Kapustin-Li pairing, one does not have to compute the Hochschild complex but rather the cyclic complex $(\text{CC}_{\bullet}(A),b+uB)$ whose differential also features Connes' operator~$B$. So instead of~\eqref{AKLtr} one must consider
\begin{align*}
(\mathcal C^2_{\text{cl}}(B_{A}), L_{Q_{1,2}})^* &\longrightarrow (\mathcal C^0(B_{A}), L_{Q_{1,2}})^* 
\longrightarrow (C^\lambda(A),b_{1,2})
\stackrel{\kappa} {\longrightarrow} (\text{CC}_{\bullet}(A),b_{1,2}+uB) \\
& \stackrel{T_{*}} {\longrightarrow} (\text{CC}_{\bullet}(A),b_{0,2}+uB)
\stackrel{\str} {\longrightarrow} (\text{CC}_{\bullet}(R),\widehat b_{0,2}+uB) \\
& \stackrel{\text{HKR}} {\longrightarrow}(\Omega^\bullet(\C^N), \D W\wedge(\,\cdot\,)+ud)  
\stackrel{\rho} {\longrightarrow} \C 
\end{align*}
where $C^\lambda(A)$ denotes Connes' complex. Again for most of the above maps it is straightforward to check that they are explicitly constructible quasi-isomorphisms. Only the maps~$\kappa$ and~$\rho$ are more interesting: $\kappa$ is a generalisation to the DG level of a result on associative algebras in~\cite{Kassel}, while one way to obtain~$\rho$ is to apply the method of homological perturbation also discussed in the next subsection. 

In the end one arrives at a homologically symplectic form $\omega\in \mathcal C^2(B_{A})$ whose constant part is the Kapustin-Li pairing. The details of this construction will appear in future work. We note that this explicit expression for~$\omega$ together with the results of the present paper allow for a significantly more efficient and general computational method to determine the effective superpotential $\mathcal W_{\text{eff}}(u,t)$ to arbitrary orders also in the bulk moduli from first principles.

\subsection{Transporting bulk deformations on-shell}\label{transbulkonshell}

After having found the solutions to the Maurer-Cartan equation describing bulk-induced deformations of the off-shell open string algebra, we are now faced with the task of transporting them on-shell. We shall do so by constructing an $L_{\infty}$-morphism
\be\label{LLmorph}
(\mathrm{Coder}({T_{A}}), [\partial, \,\cdot\,], [\,\cdot\,, \,\cdot\,]) \longrightarrow 
(\mathrm{Coder}({T_{H}}), [\widetilde{\partial}, \,\cdot\,], [\,\cdot\,, \,\cdot\,]) \, . 
\ee

A crucial observation is that $({T_{H}},\widetilde\del)$ is a deformation retract of $({T_{A}},\del)$. We will start with a general discussion of this notion on the level of complexes and of the natural $L_{\infty}$-morphism that it gives rise to. Then we will specialise to our case of open string algebras, explicitly construct the associated deformation retract data, and thus arrive at the $L_{\infty}$-morphism~\eqref{LLmorph} to transport the off-shell deformation~$\delta$ from the previous subsection to the on-shell algebra $(H,\widetilde\del)$. While we will apply it to the case of Landau-Ginzburg models, we note that our construction of the map~\eqref{LLmorph} works for arbitrary $A_{\infty}$-algebras $(A,\del)$ and their minimal models $(H,\widetilde\del)$.

\subsubsection{Deformation retractions}

A \textsl{deformation retraction} 
\begin{equation}
\xymatrix{%
(C_2, d_2) \ar@<+.5ex>[rr]^-{i}  && (C_1, d_1) \ar@<+.5ex>[ll]^-{p} \\
}%
\!\!\!\xymatrix{%
{}\ar@(ur,dr)[]^{h}
}%
\label{HDRdiagram}
\end{equation}
consists of the following data: two complexes $(C_1, d_1)$ and $(C_2, d_2)$, two maps of complexes $i : (C_2, d_2) \rightarrow (C_1, d_1)$ and $p: (C_1, d_1) \rightarrow (C_2,d_2)$, and a homotopy $h \in \mathrm{End}(C_1)$. These data are subject to the relations
\begin{align}\label{HDR}
p  i &= \mathrm{id}_{C_2} \, , \quad 
\mathrm{id}_{C_1} - i  p = d_1  h + h  d_1 \, ,
\end{align}
and we refer to $(C_{2}, d_{2})$ as the deformation retract of $(C_{1},d_{1})$. The homotopy~$h$ is said to be in \textsl{standard form} if it satisfies
$$
hh=hi=ph=0 \, .
$$ 

Given the data (\ref{HDRdiagram}), we can construct maps $M_n^1 : \mathrm{End}(C_1)[1]^{\wedge n} \rightarrow \mathrm{End}(C_2)[1]$ via
$$
M_n^1(a_1\wedge \ldots \wedge a_n) = \sum_{\sigma \in S_n} \varepsilon_{\sigma; a_1,\ldots,a_n} p a_{\sigma(1)}\, h\, a_{\sigma(2)} \ldots h a_{\sigma(n)} i
$$
for all homogeneous $a_1 ,\ldots, a_n \in \mathrm{End}(C_1)$, where $\varepsilon_{\sigma; a_1,\ldots,a_n}$ denotes the Koszul sign introduced in subsection~\ref{subsec:AinfLinf}. We will denote by~$M$ the coalgebra morphism $S_{\mathrm{End}(C_1)}\rightarrow S_{\mathrm{End}(C_2)}$ uniquely defined from the maps $M_{n}^1$. This morphism is the central ingredient of our version of the homological perturbation lemma: 

\begin{proposition}\label{HPL}
$M: (\mathrm{End}(C_1), [d_1, \,\cdot\,], [\,\cdot\,, \,\cdot\,]) \rightarrow (\mathrm{End}(C_2), [d_2, \,\cdot\,], [\,\cdot\,, \,\cdot\,])$ is an $L_{\infty}$-quasi-isomorphism. 
\end{proposition}

\begin{proof}
It is convenient to define also the collection of maps $S_n^1: \mathrm{End}(C_1)[1]^{\otimes n} \rightarrow \mathrm{End}(C_2)[1]$ by 
$$
S_n^1(a_1 \otimes \ldots \otimes a_n) = p a_{1} h a_{2} \ldots h a_{n} i \, .
$$
In fact we are going to prove that the corresponding coalgebra morphism $S$ is an $A_{\infty}$-quasi-isomorphism and hence~$M$ will be the $L_{\infty}$-morphism induced by $S$ on the commutator algebra. 

First we prove that $S$ is an $A_{\infty}$-morphism, i.\,e.
\begin{equation}\label{Smorph}
S_n^1b_n^n + S_{n-1}^1b_n^{n-1} = \widetilde{b}_1^1S_n^1 + \widetilde{b}_2^1 S_n^2
\end{equation}
where $b_1^1$ and $b_2^1$ are defined by $b_1^1(a) = [d_1, a]$ and $b_2^1(a_1 \otimes a_2) = (-1)^{\widetilde{a}_{1}}a_1a_2$ and similarly $\widetilde{b}_1^1$ and $\widetilde{b}_2^1$. For $n=1$ the condition~\eqref{Smorph} is easily checked, 
$$
\widetilde{b}_1^1S_1^1(a) = [d_2, p a i] = p [d_1, a]i = S_1^1b_1^1(a)
$$
where we have only used that $i$ and $p$ are maps of complexes. 

To prove~\eqref{Smorph} for all $n\geqslant 1$ we first compute
\begin{align*}
&\widetilde{b}_1^1S_n^1(a_1 \otimes\ldots\otimes a_n) = p [d_1, a_1h a_{2} \ldots h a_n]i\\
&= p \K{(}{)}{\sum_{k=1}^{n-1} (-1)^{\sum_{l=1}^{k-1}\widetilde{a}_i} a_1h\ldots [d_1, a_k  h]\ldots h a_n+ (-1)^{\sum_{l=1}^{n-1}\widetilde{a}_{i}}a_1 h a_{2} \ldots h [d_1, a_n] } i\\
&= S_n^1 b_n^n(a_1 \otimes\ldots \otimes a_n) - p \K{(}{)}{\sum_{k=1}^{n-1} (-1)^{\sum_{l=1}^{k}\widetilde{a}_{i}} a_1 h\ldots a_k [d_1, h] a_{k+1} \ldots h a_n} i
\end{align*}
where we have only used that $S_1^1$ is a map of complexes. Now we only have to insert the identity $[d_1, h] = \id_{C_{1}} - ip$ in the above equation to obtain the desired result. 

In order to conclude the proof we still need to show that $M_1^1$ is a quasi-isomorphism. We already know that $M_1^1$ is a map of complexes and since $(C_2, d_2)$ is a deformation retract of $(C_1, d_1)$, we are left to verify that if $a \in \mathrm{End}(C_{1})$ represents a non-trivial element in $b_1^1$-cohomology, then $p a i \neq 0$. Suppose however $p a i = 0$, then
\be\label{cohomologyEnd}
0 = i p a i p = (\id_{C_{1}} - [d_1, h])a(\id_{C_{1}} - [d_1, h])
\ee
and hence $a = b_1^1([h,a] - ha[d_1, h])$. This contradicts the assumption on~$a$.
\end{proof}

\subsubsection[Deformation retractions from $A_{\infty}$-algebras]{Deformation retractions from $\boldsymbol{A_{\infty}}$-algebras}\label{DRforcoalgebras}

The $L_{\infty}$-morphism~$M$ allows us to transport deformations of $(C_{1},d_{1})$ to deformations of $(C_{2},d_{2})$. In our case of interest these two complexes are given by $({T_{A}},\del)$ and $({T_{H}},\widetilde\del)$, respectively, and we ask for the additional property that the deformation of~$\widetilde\del$ must continue to be an $A_{\infty}$-structure. Hence we will now explain under which circumstances this is guaranteed to be the case, i.\,e.~when~\eqref{LLmorph} maps coderivations to coderivations. 

\begin{proposition}\label{HPLCoder}
Assume that
\begin{equation}\label{HDRdiagram-Ainf}
\xymatrix{%
({T_{A_{2}}}, \partial_2) \ar@<+.5ex>[rr]^-{i}  && ({T_{A_{1}}}, \partial_1) \ar@<+.5ex>[ll]^-{p} \\
}%
\!\!\!\xymatrix{%
{}\ar@(ur,dr)[]^{h}
}%
\end{equation}
is a deformation retraction where $(A_{1},\del_{1})$ and $(A_{2},\del_{2})$ are $A_{\infty}$-algebras, and~$h$ is in standard form. Then for $a_1, \ldots, a_n \in \mathrm{Coder}({T_{A_{1}}})$, we have $M_n^1(a_1 \wedge\ldots\wedge a_n) \in \mathrm{Coder}({T_{A_{2}}})$ for all $n\geqslant 1$.
\end{proposition}

\begin{proof}
Let us introduce some convenient notation: We write $\mathcal{A}=\mathrm{End}({T_{A_{1}}})$, and for an element $f \in \mathrm{End}({T_{A_{1}}})$, let $L_f$ and $R_f$ denote the left and right multiplication by $f$, respectively. Define the left and right ideals $I_L = \mathrm{Ker}(R_{i})\cap\mathrm{Ker}(R_{h})$ and $I_R = \mathrm{Ker}(L_{p})\cap\mathrm{Ker}(L_{h})$. Since~$h$ is in standard form, we have $h \in I = I_L \cap I_R$. Finally we define $\pi=ip$, $J = I_L + I_R$ and $\mathcal{B} = \mathrm{lin}_{\C}(\mathrm{id}_{T_{A_1}} - \pi)$.

Now we will show that in fact $S_n^1(a_1 \otimes \ldots \otimes a_n)$ is a coderivation. Denote $\Lambda_n = a_1 h a_2 \ldots h a_n$ and assume without loss of generality that $a_1,\ldots, a_n$ are homogeneous. The crucial observation is that $\Delta \Lambda_n$ admits the decomposition (proved in appendix~\ref{HPLforcoalgebras})
\begin{equation}
\Delta \Lambda_n = ((\mathrm{id}_{{T_{A_{1}}}} + \mathcal{B})\otimes \Lambda_n + \Lambda_n \otimes (\mathrm{id}_{{T_{A_{1}}}} + \mathcal{B})+ J \otimes \mathcal{A} + \mathcal{A} \otimes J) \Delta \, ,
\label{decomposition}
\end{equation}
where we use a short-hand notation where e.\,g.~``$J\otimes \mathcal A$'' means ``some element in $J\otimes \mathcal A$''. 
It then follows that $S_n^1(a_1\otimes \ldots \otimes a_n) = p \Lambda_n i$ satisfies
$$
\Delta p \Lambda_n i = (p \Lambda_n i\otimes \mathrm{id}_{{T_{A_{2}}}} + \mathrm{id}_{{T_{A_{2}}}} \otimes p \Lambda_n i)\Delta \, ,
$$
which says that $S_n^1(a_1\otimes \ldots \otimes a_n) $ is a coderivation. 
\end{proof}

We have thus proved that~$M$ continues to be an $L_{\infty}$-morphism when restricted to coderivations. However, it will then generically no longer be a quasi-isomorphism (as we discuss in appendix~\ref{appontooff}).

\subsubsection{On-shell bulk-induced deformations}

Proposition \ref{HPLCoder} enables us, given a deformation retraction of $A_{\infty}$-algebras (\ref{HDRdiagram-Ainf}), to transport deformations of $\partial_1$ to deformations of $\partial_2$. To accomplish our aim to construct bulk-deformed open topological string theories for Landau-Ginzburg models, we are now left to specify the deformation retract data
\begin{equation}
\label{HDRdiagram-ourXX}
\xymatrix{%
({T_{H}}, \widetilde{\partial}) \ar@<+.5ex>[rr]^-{F}  && ({T_{A}}, \partial) \ar@<+.5ex>[ll]^-{\bar F} \\
}%
\!\!\!\xymatrix{%
{}\ar@(ur,dr)[]^{U}
}%
\end{equation}
paying attention to the condition that the homotopy~$U$ be in standard form. In writing~\eqref{HDRdiagram-ourXX} we have already revealed that in the case at hand the inclusion map is given by the minimal model morphism $F:(H,\widetilde\del)\rightarrow (A,\del)$ of proposition~\ref{minmodthm}. It remains to find its homotopy inverse~$\bar F$ and the homotopy~$U$ itself. This is achieved by the next proposition which constructs~$\bar F$ and~$U$ explicitly. 

\begin{proposition}
\label{HDR}
For any $A_{\infty}$-algebra $(A,\del)$, there exist a unique colagebra morphism~$\bar F$ and a map~$U$ that make~\eqref{HDRdiagram-ourXX} a deformation retraction and satisfy the conditions
\begin{align}
\Delta U &= \frac{1}{2}(U \otimes (\id_{T_{A}} + F\bar{F}) + (\id_{T_{A}} + F \bar{F}) \otimes U)\Delta \, , \label{DeltaUdef}\\
U_1^1 &= G\, , \nonumber \\
U_n^1 &= -G\partial_2^1U_n^2\textrm{ for }n \geqslant 2\, , \label{Un1} \\
\bar{F}_1^1 &= \pi_H\, , \nonumber \\
\bar{F}_n^1 &= -\pi_H([\partial, U])_n^1 = -\pi_H \partial_2^1U_n^2\, . \label{Ftilde1n}
\end{align}
Moreover, these conditions allow for an explicit construction of $\bar F,U$, and~$U$ is in standard form. 
\end{proposition}

The rather technical proof of the above proposition can be found in appendix~\ref{UFtildeproof}. 

\medskip

Now we have arrived at the point to put together all the results obtained in this section. We apply the $L_{\infty}$-morphism~$M$ of proposition~\ref{HPL} in conjunction with the deformation retraction of proposition~\ref{HDR}. Recall that in the previous subsection we found that for Landau-Ginzburg models off-shell deformations $\delta\in\text{Coder}({T_{A}})$ are precisely the bulk-induced coderivations determined by
\be\label{deltaphi}
\delta_{0}^1 = \sum_{i} t_{i}\phi_{i}
\ee
where $\{\phi_{i}\}$ is a basis of the bulk space $\Jac(W)$, and~$t_{i}$ are the associated moduli. By proposition~\ref{prop:MCsolutions} we can use~$M$ to map~$\delta$ to deformations~$\widetilde\delta$ of the on-shell open string algebra $(H,\widetilde\del)$, and proposition~\ref{HPLCoder} ensures that $\widetilde\del+\widetilde\delta$ indeed encodes an $A_{\infty}$-structure. Thus our final result is the following.

\begin{theorem}\label{thm:bulkdeformations}
The bulk-induced deformations of the on-shell Landau-Ginzburg open string algebra $(H,\widetilde\del)$ are given by
\be\label{deltatilde}
\widetilde{\delta} = \sum_{n \geqslant 1}\frac{1}{n!}M_n^1(\delta^{\wedge n}) = \sum_{n \geqslant 1}\bar{F} (\delta U)^n \delta F = \bar F(\id_{{T_{A}}} - \delta U)^{-1} \delta F \, .
\ee
By substituting~\eqref{deltaphi} together with the concrete formulas for $F,\bar F$ and~$U$ into~\eqref{deltatilde}, one obtains explicit expressions for bulk-deformed $A_{\infty}$-products on~$H$.
\end{theorem}

To make sure that the bulk moduli dependent $A_{\infty}$-structure encoded in $\widetilde\del+\widetilde\delta$ immediately describes all amplitudes of bulk-deformed open topological string theory, it has to be shown that also the deformed $A_{\infty}$-products are cyclic with respect to the Kapustin-Li pairing. 

Let us for the moment restrict to those bulk fields that ``are seen by the open TFT of the brane~$D$'', i.\,e.~those $\phi\in\Jac(W)$ that are not mapped to zero by the bulk-boundary map. We denote an off-shell deformation that arises from such a bulk field by $\delta_{Z}$. Then the on-shell deformation $\widetilde\delta_{Z}$ takes a particularly simple form in our approach: 
$$
\widetilde\delta_{Z} =  \bar F(\id_{{T_{A}}} - \delta_{Z} U)^{-1} \delta_{Z} F = \bar F \delta_{Z} F= \delta_{Z} \, ,
$$
where we have used the fact that the image of~$F$ only consists of tensor powers of elements in~$H$ and the complement of $\text{Ker}([D,\,\cdot\,])$ (see~\eqref{AHBL}), and that~$U$ acts as~$G$ on one tensor factor of each summand. Furthermore, since $\delta_{Z}$ is a pure coordinate coderivation, it follows that the Lie derivative of the flat part of the symplectic form~$\omega$ mentioned in subsections~\ref{oTSTforLG} and~\ref{omegaconstruction} with respect to $\delta_{Z}^\vee$ vanishes. As explained in detail in~\cite{c0904.0862}, this means that the $A_{\infty}$-structure $\widetilde\del+\widetilde\delta_{Z}$ is indeed cyclic. 

The rigidity of the methods used in transporting bulk deformations on-shell may suggest the cyclicity of $\widetilde\del+\widetilde\delta$ also in the general case when the off-shell bulk deformation is not of the form $\delta_{Z}$. However, then the abstract manipulation of $\widetilde\delta$ in~\eqref{deltatilde} is more difficult, and at this point we have no proof that $\widetilde\del+\widetilde\delta$ is cyclic. 
But even if this were not the case, if there are no obstructions in principle one can do on-shell tadpole-cancellation as discussed in~\cite{hll0402}, i.\,e.~construct a weak $A_\infty$-isomorphism from $(H, \widetilde\del+\widetilde\delta)$ to a non-curved $A_\infty$-algebra. As opposed to the case of off-shell tadpole-cancellation (see the proof of lemma~\ref{Tmaplemma}) such a map may not exist, and its conceptual understanding is incomplete; it exists precisely iff there are no obstructions. Then the approach of~\cite{c0904.0862} is applicable and one can obtain the correct bulk-deformed, tadpole-cancelled and cyclic $A_{\infty}$-structure. This is in particular true for the special brane which is the compact generator of the category all matrix fatorisations, and hence by a standard argument the Calabi-Yau $A_\infty$-structure can be lifted to the full D-brane category in this case (see~\cite{d0904.4713}).

\section{Conclusion}

In the present paper we have obtained a first-principle derivation of the bulk-induced deformations of all open topological string amplitudes arising from affine B-twisted Landau-Ginzburg models. Our approach is a conceptual step forward in that it provides a \textsl{constructive} method to compute amplitudes, instead of trying to solve their defining $A_{\infty}$-constraints only by brute force. 

Our results will also be of practical use in the following sense. Together with the methods of~\cite{c0904.0862} or the direct construction of the off-shell Kapustin-Li pairing as sketched in subsection~\ref{omegaconstruction}, our results enable us to perturbatively compute the effective D-brane superpotential $\mathcal W_{\text{eff}}$ at tree level prior to tadpole-cancellation, to all orders in both boundary and bulk moduli, and for all branes in all Landau-Ginzburg models. 

The main line of thought that permeates this work consists in making rigorous the link between the DG Lie algebra of bulk fields and the DG Lie algebra governing deformations of open string amplitudes, i.\,e.~of the underlying Calabi-Yau $A_{\infty}$-structure. Our solution to the problem is the explicit construction of an $L_{\infty}$-morphism that transports pure bulk deformations to the boundary sector. This map is the composition of an $L_{\infty}$-quasi-isomorphism from the off-shell bulk sector to the off-shell boundary sector, and of an $L_{\infty}$-morphism transporting off-shell boundary deformations on-shell. We saw that the former is a ``weak'' version of deformation quantisation, while the latter can be viewed as homological perturbation in its $L_{\infty}$-incarnation. 

\medskip

There are many interesting directions of research that are opened by our construction, and we shall name a few. One is rather immediate: in string theory it is orbifolded Landau-Ginzburg models that are particularly relevant, and one should extend our methods for bulk deformations to orbifolds. A related question is that of the actual perturbative computation of $\mathcal W_{\text{eff}}$ in concrete examples. In the present paper it was our aim to address conceptual questions and give a completely general prescription to compute all bulk-deformed amplitudes in arbitrary affine Landau-Ginzburg models. But since our construction in theorem~\ref{thm:bulkdeformations} is entirely explicit, it will be straightforward to implement it on a computer. We leave such an implementation and the computation of examples for future work. As mentioned before, the efficiency of the algorithm will be much improved by the explicit formula for the off-shell Kapustin-Li pairing as sketched in subsection~\ref{omegaconstruction}. 

A more conceptual matter is the following. While our map~\eqref{L1} transporting off-shell bulk deformations to off-shell boundary deformations is an $L_{\infty}$-quasi-isomorphism, the map~\eqref{L2} to on-shell boundary deformations is not necessarily a quasi-isomorphism. As a consequence there might be on-shell boundary deformations that are not bulk-induced. This may be related to the possible non-uniqueness of boundary sectors for a given bulk sector (in the case of non-rational theories), as well as the different roles played by Hochschild cohomology of the first and second kind. 

Finally, we mention a question that goes beyond affine Landau-Ginzburg models. Our results appear to be the first coherent $L_{\infty}$-formulation of and solution to the problem of bulk-induced deformations in such models. It is natural to ask to what extent and how they can be generalised to non-affine Landau-Ginzburg models and even more general open topological string theories. Recall that our bulk-to-boundary transport splits into the two maps~\eqref{L1} and~\eqref{L2}, the second of which is a model-independent construction. Hence the interesting question in going beyond affine Landau-Ginzburg models is how to generalise the map~\eqref{L1}. We expect that part of the general answer is that the left-hand side of~\eqref{L1} will continue to be the \textsl{off-shell} bulk space which always has a DG Lie algebra structure, just like the right-hand side of~\eqref{L1}. Any such generalisation would at the same time also give rise to a new variant of deformation quantisation.

\subsubsection*{Acknowledgements}

We thank 
    Patrick B\"ohl, 
    Ilka Brunner, 
    Manfred Herbst, 
    Calin Lazaroiu, 
    Wolfgang Lerche, 
    Daniel Murfet
    and Masoud Soroush 
for discussions.

\appendix

\section{Details of some proofs}

\subsection{A spectral sequence}\label{appss}

Here we present the spectral sequence computation used in the proof of proposition~\ref{wDQ}. First we note that $(\mathrm{Coder}({T_{R}}), [\widehat{\partial}_0 + \widehat{\partial}_2, \,\cdot\,])$ is a mixed complex: $[\partial_0, \,\cdot\,]$ decreases tensor degree by $1$ and increases tilde degree by~1, while $[\partial_2, \,\cdot\,]$ increases tensor degree by $1$ and increases tilde degree by 1. In order to construct a bicomplex, we organise tensor and tilde degrees as follows
$$
\xymatrix{%
 & \vdots  & \vdots &\vdots  & \vdots  &  \\
\cdots \ar[r]^-{-d_0} & C_{0+s}^3 \ar[u]_{d_2} \ar[r]^-{-d_0} & C_{1+s}^2 \ar[u]_{d_2} \ar[r]^-{-d_0} & C_{0+s}^1 \ar[u]_{d_2} \ar[r]^-{-d_0} & C_{1+s}^0 \ar[u]_{d_2} \ar[r]^-{-d_0} & \ldots\\
\cdots \ar[r]^-{d_0} &C_{1+s}^2 \ar[u]_{d_2} \ar[r]^-{d_0} & C_{0+s}^1 \ar[u]_{d_2} \ar[r]^-{d_0} & C_{1+s}^0 \ar[u]_{d_2} \ar[r]^-{d_0} & 0 \ar[u]_{d_2} \ar[r]^-{d_0} & \ldots\\
\cdots \ar[r]^-{-d_0} &C_{0+s}^1 \ar[u]_{d_2} \ar[r]^-{-d_0} & C_{1+s}^0 \ar[u]_{d_2} \ar[r]^-{-d_0} & 0 \ar[u]_{d_2} \ar[r]^-{-d_0} & 0 \ar[u]_{d_2} \ar[r]^-{-d_0} & \ldots\\
\cdots \ar[r]^-{d_0} &C_{1+s}^0 \ar[u]_{d_2} \ar[r]^-{d_0} & 0 \ar[u]_{d_2} \ar[r]^-{d_0} & 0 \ar[u]_{d_2} \ar[r]^-{d_0} & 0 \ar[u]_{d_2} \ar[r]^-{d_0} & \ldots
\\}%
$$
where $s \in \{0,1\}$ and $C_m^n$ denotes the subspace of $\mathrm{Coder}({T_{R}})$ of tilde degree $m$ and tensor degree $n$, and we write $d_0 = [\widehat{\partial}_0, \,\cdot\,]$ and $d_2 = [\widehat{\partial}_2, \,\cdot\,]$. 

We choose the first page of the spectral sequence computing $[\widehat\del_{0}+\widehat\del_{2},\,\cdot\,]$-cohomology of $\Coder({T_{R}})$ to be the cohomology of $d_2$, which is given by replacing $C_{\bullet}^{\bullet}$ above with the image of $K_1^1$ of appropriate degrees. Since $K_1^1 l_1^1 = c_1^1 K_1^1$, for $s = 1$ the second page vanishes, while for $s=0$ it is
$$
\xymatrix{%
 & \vdots  & \vdots &\vdots  & \vdots  &  \\
\cdots \ar[r]^-{-d_0} & 0 \ar[u]_{d_2} \ar[r]^-{-d_0} & 0 \ar[u]_{d_2} \ar[r]^-{-d_0} & 0 \ar[u]_{d_2} \ar[r]^-{-d_0} & \Jac(W) \ar[u]_{d_2} \ar[r]^-{-d_0} & \ldots\\
\cdots \ar[r]^-{d_0} & 0 \ar[u]_{d_2} \ar[r]^-{d_0} & 0 \ar[u]_{d_2} \ar[r]^-{d_0} & \Jac(W) \ar[u]_{d_2} \ar[r]^-{d_0} & 0 \ar[u]_{d_2} \ar[r]^-{d_0} & \ldots\\
\cdots \ar[r]^-{-d_0} & 0 \ar[u]_{d_2} \ar[r]^-{-d_0} & \Jac(W) \ar[u]_{d_2} \ar[r]^-{-d_0} & 0 \ar[u]_{d_2} \ar[r]^-{-d_0} & 0 \ar[u]_{d_2} \ar[r]^-{-d_0} & \ldots\\
\cdots \ar[r]^-{d_0} & \Jac(W) \ar[u]_{d_2} \ar[r]^-{d_0} & 0 \ar[u]_{d_2} \ar[r]^-{d_0} & 0 \ar[u]_{d_2} \ar[r]^-{d_0} & 0 \ar[u]_{d_2} \ar[r]^-{d_0} & \ldots
\\}%
$$
Here for degree reasons the spectral sequence degenerates, yielding the desired result. 

We mention that instead of computing the Hochschild cohomology of $(R,\widehat{\partial}_0 + \widehat{\partial}_2)$ we could also compute that of $(A,\del_{1}+\del_{2})$ which by remark~\ref{HHremark} and the existence of the weak isomorphism~$T$ in the proof of lemma~\ref{Tmaplemma} is isomorphic to $\HH^\bullet(R,\widehat{\partial}_0 + \widehat{\partial}_2)$. The analogous spectral sequence in the case of $(A,\del_{1}+\del_{2})$ is more involved and degenerates only at the third page.

\subsection{Homological perturbation for coalgebras}\label{HPLforcoalgebras}

We continue the proof of proposition~\ref{HPLCoder}. We will establish~\eqref{decomposition} by induction. For this it is convenient to consider a sequence $\{a_i\}_{i \in \N} \subset \mathrm{Coder}({T_{A_{1}}})$. Then we have
\begin{align*}
\Delta\Lambda_{n+1} & = \Delta \Lambda_n h a_{n+1}\\
&= ((\mathrm{id}_{T_{A_1}} + \mathcal{B}) \otimes \Lambda_n + \Lambda_n \otimes (\mathrm{id}_{T_{A_1}} + \mathcal{B}) + J \otimes \mathcal{A} + \mathcal{A} \otimes J)  \\
&\quad \cdot \frac{1}{2}\Big(h\otimes(\mathrm{id}_{T_{A_1}} + \pi) + (\mathrm{id}_{T_{A_1}}+\pi) \otimes h \Big) (a_{n+1} \otimes \mathrm{id}_{T_{A_1}} + \mathrm{id}_{T_{A_1}} \otimes a_{n+1})\Delta \, .
\end{align*}
The computation naturally splits into two steps, one involving the summand $(\mathrm{id}_{T_{A_1}} + \mathcal{B})\otimes \Lambda_n + \Lambda_n \otimes (\mathrm{id}_{T_{A_1}} + \mathcal{B})$ and the other involving the term $J \otimes \mathcal{A} + \mathcal{A} \otimes J$ in the first factor on the right-hand side above. For the first piece we have
\begin{align*}
&((\mathrm{id}_{T_{A_1}} + \mathcal{B})\otimes \Lambda_n + \Lambda_n \otimes (\mathrm{id}_{T_{A_1}} + \mathcal{B}))\\
&\quad\cdot \Big(h\otimes\frac{1}{2}(\mathrm{id}_{T_{A_1}} + \pi) + \frac{1}{2}(\mathrm{id}_{T_{A_1}}+\pi) \otimes h)(a_{n+1} \otimes \mathrm{id}_{T_{A_1}} + \mathrm{id}_{T_{A_1}} \otimes a_{n+1}\Big)\\
&= \Big((\mathrm{id}_{T_{A_1}} + \mathcal{B}) \otimes \Lambda_n + \Lambda_n \otimes (\mathrm{id}_{T_{A_1}} + \mathcal{B})\Big) \Big(ha_{n+1}\otimes \frac{1}{2}(\mathrm{id}_{T_{A_1}} + \pi) \\
&\quad + h \otimes \frac{1}{2}(\mathrm{id}_{T_{A_1}} + \pi)a_{n+1} + (-1)^{\widetilde{a}_{n+1}}\frac{1}{2}(\mathrm{id}_{T_{A_1}}+\pi)a_{n+1}\otimes h \\
& \quad + \frac{1}{2}(\mathrm{id}_{T_{A_1}}+\pi)\otimes ha_{n+1}\Big)\\
&= I_R \otimes \mathcal{A} + I \otimes \mathcal{A} + \mathcal{A}\otimes I_L + (\mathrm{id}_{T_{A_1}} + \mathcal{B})\otimes \Lambda_n h a_{n+1}\\
&\quad+ \Lambda_n h a_{n+1} \otimes (\mathrm{id}_{T_{A_1}} + \mathcal{B}) + I_L \otimes \mathcal{A} + \mathcal{A} \otimes I + \mathcal{A}\otimes I_R \, ,
\end{align*}
where in the last step we have used that $L_h \mathcal{B} \subset I_L$. This is true because $h (\mathrm{id}_{T_{A_1}} -\pi) = h \partial h$. We have thus proved that the first piece in the computation is of the desired form. For the second piece we have
\begin{align*}
&(J \otimes \mathcal{A} + \mathcal{A} \otimes J) \Big(ha_{n+1}\otimes \frac{1}{2}(\mathrm{id}_{T_{A_1}} + \pi) + h \otimes \frac{1}{2}(\mathrm{id}_{T_{A_1}} + \pi)a_{n+1}\\
&\quad + (-1)^{\widetilde{a}_{n+1}}\frac{1}{2}(\mathrm{id}_{T_{A_1}}+\pi)a_{n+1}\otimes h + \frac{1}{2}(\mathrm{id}_{T_{A_1}}+\pi)\otimes ha_{n+1}\Big)\\
&= I_R \otimes \mathcal{A} + I \otimes \mathcal{A} + \mathcal{A}\otimes I_L + J \otimes \mathcal{A}+ \mathcal{A}\otimes J + I_L \otimes \mathcal{A} + \mathcal{A} \otimes I + \mathcal{A} \otimes I_R\end{align*}
which is again of the desired form.

\subsection[More on the $L_{\infty}$-morphism $M$]{More on the $\boldsymbol{L_{\infty}}$-morphism $\boldsymbol{M}$}\label{appontooff}

This appendix supplements subsection~\ref{DRforcoalgebras}. We give a brief explanation of why $M_1^1: \mathrm{Coder}(T_{A_1}) \rightarrow \mathrm{Coder}(T_{A_2})$ is generically not a quasi-isomorphism. For simplicity we denote $C_{A_i} = \mathrm{Coder}(T_{A_i})$ for $i \in\{1,2\}$ and define $g = M_1^1 = p(\,\cdot\,)i$. Consider the short exact sequence 
\begin{equation}
\xymatrix{
0 \ar@<+.0ex>[r] & \mathrm{Ker}(g) \ar@<+.5ex>[r]^-{f} & C_{A_1} \ar@<+.5ex>[l]^-{r} \ar@<+.5ex>[r]^{g} & C_{A_2} \ar@<+.5ex>[l]^-{s} \ar@<+.0ex>[r]  & 0
}\label{shortexseq}
\end{equation}
where $f$ denotes inclusion. As a sequence of modules, (\ref{shortexseq}) is split exact with left and right inverses $r$ and $s$ respectively given by $r(\phi)_n^1 = (\phi - \pi\phi\pi)_n^1$ and $s(\psi)_n^1 = (i \psi p)_n^1$ for $\phi \in C_{A_1}$ and $\psi \in C_{A_2}$. If we view $\mathrm{Ker}(g)$, $C_{A_1}$ and $C_{A_2}$ as complexes with the appropriate Hochschild differentials, $f$and $g$ are promoted to maps of complexes and one can study the cohomology of $C_{A_2}$, i.\,e.~the Hochschild cohomology of $(A_2, \partial_2)$, by analysing the long exact sequence of cohomology groups
\begin{equation}
\xymatrix{
\ldots \ar@<+.0ex>[r] & H(\mathrm{Ker}(g)) \ar@<+.0ex>[r] & H(C_{A_1}) \ar@<+.0ex>[r] & H(C_{A_2})\ar@<+.0ex>[r]^-{\delta} & H(\mathrm{Ker}(g))\ar@<+.0ex>[r] & \ldots
}\label{longexseq}
\end{equation}
where the coboundary map $\delta$ is given by $\delta(\psi)= r([\partial_1, s(\psi)])$ for $\psi \in C_{A_2}$.

If $r$, $s$ were maps of complexes, \eqref{shortexseq} would be promoted to a split exact sequence of complexes, and in that case \eqref{longexseq} would reduce to a short exact sequence. It is readily seen however that $r$ and $s$ are not necessarily maps of complexes, and the following modification must be made. In general we can truncate \eqref{longexseq} as
\begin{equation}
\xymatrix{
0 \ar@<+.0ex>[r] & H(\mathrm{Ker}(g))/\mathrm{Im}(\delta) \ar@<+.0ex>[r] & H(C_{A_1}) \ar@<+.0ex>[r] & H(C_{A_2})\ar@<+.0ex>[r] & \mathrm{Im}(\delta) \ar@<+.0ex>[r] & 0 \, .
}\label{trunc1exseq}
\end{equation}
By the split property of \eqref{shortexseq} we observe that
$$
\mathrm{Im}(\delta) = (\mathrm{Im}([\partial_1, \,\cdot\,]) \cap \mathrm{Ker}(g))/\mathrm{Im}([\partial_1,\,\cdot\,])|_{\mathrm{Ker}(g)} \, ,
$$
and hence \eqref{trunc1exseq} becomes
$$
\xymatrix{
0 \ar@<+.0ex>[r] & H(C_{A_1}) \cap \mathrm{Ker}(g) \ar@<+.0ex>[r] & H(C_{A_1}) \ar@<+.0ex>[r] & H(C_{A_2}) \ar@<+.0ex>[r] & \mathrm{Im}(\delta) \ar@<+.0ex>[r] & 0
}
$$
Generically this does not simplify, in contrast to the case of the complexes of endomorphisms $\mathrm{End}(T_{A_i})$. Here the corresponding inverse maps $r$ and $s$ are maps of complexes and the long exact sequence reduces to
$$
\xymatrix{
0 \ar@<+.0ex>[r] & H(\mathrm{End}(T_{A_1})) \cap \mathrm{Ker}(g) \ar@<+.0ex>[r] & H(\mathrm{End}(T_{A_1})) \ar@<+.0ex>[r] & H(\mathrm{End}(T_{A_2})) \ar@<+.0ex>[r] & 0 \, .
}
$$
The computation in (\ref{cohomologyEnd}) then shows that $H(\mathrm{End}(T_{A_1}))\cap \mathrm{Ker}(g) = 0$ and we recover $H(\mathrm{End}(T_{A_1})) \cong  H(\mathrm{End}(T_{A_2}))$, i.\,e.~$M_1^1 : \mathrm{End}(T_{A_1}) \rightarrow  \mathrm{End}(T_{A_2})$ is a quasi-isomorphism.

\subsection{The off-shell to on-shell deformation retraction}\label{UFtildeproof}

Here we will prove proposition~\ref{HDR}. That $\bar{F}$ must be of the form~\eqref{Ftilde1n} follows by applying $\pi_H$ to
\begin{equation}
\mathrm{id}_{{T_{A}}} - F\bar{F} = [\partial, U] \, . 
\label{homotopyour}
\end{equation}
$\bar{F}$ is then automatically an $A_{\infty}$-morphism by
$$
F(\widetilde{\partial}\bar{F} - \bar{F}\partial) = [\partial, F\bar{F}]= -[\partial, \partial U + U \partial] = 0
$$
and the injectivity of $F$. Next we show that~\eqref{DeltaUdef} is compatible with (\ref{homotopyour}):
\begin{align*}
&\Delta (\id_{{T_{A}}} - F\bar{F}) = \Delta (\partial U + U \partial)\\
&= \frac{1}{2}(\partial \otimes \id_{{T_{A}}} + \id_{{T_{A}}} \otimes \partial)(U \otimes (\id_{{T_{A}}} + F\bar{F}) + (\id_{{T_{A}}} + F \bar{F}) \otimes U) \\
& \quad + \frac{1}{2}(U \otimes (\id_{{T_{A}}} + F\bar{F}) + (\id_{{T_{A}}} + F \bar{F}) \otimes U)(\partial \otimes \id_{{T_{A}}} + \id_{{T_{A}}} \otimes \partial)\\
&= \frac{1}{2}(\id_{{T_{A}}} - F\bar{F})\otimes (\id_{{T_{A}}} + F\bar{F}) + \frac{1}{2}(\id_{{T_{A}}} + F\bar{F})\otimes(\id_{{T_{A}}} - F\bar{F})\\ 
&\quad + [\partial, (\id_{{T_{A}}} + F\bar{F})] \otimes U - U \otimes [\partial , (\id_{{T_{A}}} + F\bar{F})]\\
&= (\id_{{T_{A}}} - F\bar{F}\otimes F\bar{F})\Delta = \Delta (\id_{{T_{A}}} - F\bar{F}) 
\end{align*}
where in the penultimate step we made use of the fact that $[\partial, F\bar{F}] = 0$ and in the last step we used the fact that $F\bar{F}$ is a coalgebra morphism. This calculation shows that if we chose $U_n^1$ appropriately, condition~\eqref{DeltaUdef} ensures that $U$ is a solution of~\eqref{homotopyour}. 

Inspection of~\eqref{homotopyour} reveals that $U_1^1$ must be a homotopy for $\partial_1^1$, and we can therefore choose $U_1^1 = G$. For $n\geqslant 2$ we observe that $U_{n}^1$ of the form~\eqref{Un1} satisfies
\begin{equation}
\pi_B (\partial U + U \partial)_n^1 = 0 \, ,
\end{equation}
and moreover
\begin{equation}
\pi_L (\partial U + U \partial)_n^1 = - \pi_L (F \bar{F})_n^1
\label{piLeq}
\end{equation}
holds for $n = 2$. In order to show that this is also true  for $n > 2$ we proceed by induction. We start by substituting~\eqref{Un1} into the left-hand side of \eqref{piLeq} to obtain
$$
-G\partial_2^1(\partial_2^2U_n^2 + U_{n-1}^2 \partial_n^{n-1} + U_n^2 \partial_n^n) = -G\partial_2^1(\partial U + U \partial)_n^2
= G\partial_2^1(F\bar{F})_n^2\\
= -\pi_L(F\bar{F})_n^1
$$
where in the first equality we used the associativity of $\partial$, i.\,e.~$\partial_2^1\partial_3^2 = 0$. The second equality is the induction step that is well-defined due to~\eqref{DeltaUdef}, while in the third equality we used $[\partial, F\bar{F}] = 0$ and the definition of~$F$ from~\ref{minmodthm}. 

We are thus left to verify that $\bar{F}$ is a left inverse of~$F$ and that~$U$ is in standard form. First we give the explicit recursive formulas
\begin{align}
U_n^1 & = -\frac{1}{2}G\partial_2^1 \K{(}{)}{\sum_{l = 1}^{n-1}(U_l^1 \otimes (\id_{{T_{A}}} + F\bar{F})_{n-l}^1 + (\id_{{T_{A}}} + F\bar{F})_{n-l}^1\otimes U_l^1 } \, , \nonumber \\
\bar{F}_n^1 & = -\frac{1}{2}\pi_H\partial_2^1 \K{(}{)}{\sum_{l = 1}^{n-1}(U_l^1 \otimes (\id_{{T_{A}}} + F\bar{F})_{n-l}^1 + (\id_{{T_{A}}} + F\bar{F})_{n-l}^1\otimes U_l^1 } \, . \label{tildeF}
\end{align}
The maps $U_n^m$ for $m>1$ are then completely determined by repeated application of the coproduct $\Delta$.
Now we show that $\bar{F}F = \id_{{T_{A}}}$. Since $\bar{F}F$ is a coalgebra morphism, we only need to consider the subset of equations
\begin{equation}
(\bar{F}F)_n^1 = (\mathrm{id}_{{T_{H}}})_n^1 \, .
\end{equation}
Clearly the above is satisfied for $n=1$. For $n > 1$, $\bar{F}_1^kF_k^n $ vanishes at $k=1$ due to $\mathrm{Im}(F_n^1) \subset L \subset \mathrm{Ker}(\pi_H)$. While for $k > 1$ it vanishes because from (\ref{tildeF}) we see that each summand in $\bar{F}_n^1$ has at least one factor proportional to $G$. However each tensor factor of each summand is in $\mathrm{Im}(F_k^1) \subset H \oplus L = \mathrm{Ker}(G)$.

That $U$ is in standard form, i.\,e.~$\bar{F}U = 0$, $UF = 0$ and $UU = 0$, follows from an argument in direct analogy to the above proof for $\bar{F}$.

The above proof is easily extended to the case where $(A, \partial)$ is an arbitrary $A_{\infty}$-algebra by replacing the formula for $U_n^1$ with $U_n^1 = -G\K{(}{)}{\sum_{k=2}^{n}\partial_k^1U_n^k}$ from which it follows that $\bar{F}_n^1 = -\pi_H\K{(}{)}{\sum_{k=2}^{n}\partial_k^1U_n^k}$.


\begin{thebibliography}{10}

\bibitem{ab0909.2245}
M.~Aganagic and C.~Beem, \textsl{The {G}eometry of {D}-{B}rane
  {S}uperpotentials}, \href{http://arxiv.org/abs/0909.2245}{[arXiv:0909.2245]}.

\bibitem{ahjmms0909.1842}
M.~Alim, M.~Hecht, H.~Jockers, P.~Mayr, A.~Mertens, and M.~Soroush, \textsl{Hints
  for {O}ff-{S}hell {M}irror {S}ymmetry in type {II}/{F}-theory
  {C}ompactifications},
  \href{http://arxiv.org/abs/0909.1842}{[arXiv:0909.1842]}.

\bibitem{ahmm0901.2937}
M.~Alim, M.~Hecht, P.~Mayr, and A.~Mertens, \textsl{Mirror {S}ymmetry for {T}oric
  {B}ranes on {C}ompact {H}ypersurfaces}, JHEP \textbf{0909} (2009), 126,
  \href{http://arxiv.org/abs/0901.2937}{[arXiv:0901.2937]}.

\bibitem{af0506041}
P.~S. Aspinwall and L.~M. Fidkowski, \textsl{Superpotentials for {Q}uiver {G}auge
  {T}heories}, JHEP \textbf{0610} (2006), 047,
  \href{http://arxiv.org/abs/hep-th/0506041}{[hep-th/0506041]}.

\bibitem{a1998}
M.~Atiyah, \textsl{Topological {Q}uantum {F}ield {T}heory}, Publications
  Math\'ematiques de l'IH\'ES \textbf{68} (1988), 175--186,
  \href{http://archive.numdam.org/ARCHIVE/PMIHES/PMIHES_1988__68_/PMIHES_1988_%
_68__175_0/PMIHES_1988__68__175_0.pdf}{Numdam archive}.

\bibitem{bbg0704.2666}
M.~Baumgartl, I.~Brunner, and M.~R. Gaberdiel, \textsl{D-brane superpotentials
  and {RG} flows on the quintic}, JHEP \textbf{0707} (2007), 061,
  \href{http://arxiv.org/abs/0704.2666}{[arXiv:0704.2666]}.

\bibitem{bbs1007.2447}
M.~Baumgartl, I.~Brunner, and M.~Soroush, \textsl{D-brane {S}uperpotentials:
  {G}eometric and {W}orldsheet {A}pproaches}, Nucl.~Phys.~B \textbf{843} (2011), 602--637,
  \href{http://arxiv.org/abs/1007.2447}{[arXiv:1007.2447]}.

\bibitem{bhls0305}
I.~Brunner, M.~Herbst, W.~Lerche, and B.~Scheuner, \textsl{Landau-{G}inzburg
  {R}ealization of {O}pen {S}tring {TFT}}, JHEP \textbf{0611} (2003), 043,
  \href{http://www.arxiv.org/abs/hep-th/0305133}{[hep-th/0305133]}.

\bibitem{c0904.0862}
N.~Carqueville, \textsl{Matrix factorisations and open topological string
  theory}, JHEP \textbf{0907} (2009), 005,
  \href{http://arxiv.org/abs/0904.0862}{[arXiv:0904.0862]}.

\bibitem{cqv0912.4699}
N.~Carqueville and A.~Quintero~V\'elez, \textsl{Remarks on quiver gauge theories
  from open topological string theory}, JHEP \textbf{1003} (2010), 129,
  \href{http://arxiv.org/abs/0912.4699}{[arXiv:0912.4699]}.

\bibitem{cf9902090}
A.~S. Cattaneo and G.~Felder, \textsl{A path integral approach to the
  {K}ontsevich quantization formula}, Commun. Math. Phys. \textbf{212} (2000),
  591--611, \href{http://arxiv.org/abs/math/9902090}{[math.QA/9902090]}.

\bibitem{c0412149}
K.~J. Costello, \textsl{Topological conformal field theories and {C}alabi-{Y}au
  categories}, Adv. in Math. \textbf{210} (2007),
  \href{http://arxiv.org/abs/math.QA/0412149}{[math.QA/0412149]}.
  
\bibitem{c0403266}
M.~Crainic, \textsl{On the perturbation lemma, and deformations},
  \href{http://arxiv.org/abs/math/0403266}{[math.AT/0403266]}.

\bibitem{ct1007.2679}
A.~{C\u ald\u araru} and J.~Tu, \textsl{Curved {$A_{\infty}$} algebras and
  {L}andau-{G}inzburg models},
  \href{http://arxiv.org/abs/1007.2679}{[arXiv:1007.2679]}.

\bibitem{d9703136}
R.~Dijkgraaf, \textsl{Les {H}ouches {L}ectures on {F}ields, {S}trings and
  {D}uality}, \href{http://arxiv.org/abs/hep-th/9703136}{[hep-th/9703136]}.

\bibitem{d1005.2779}
M.~R. Douglas, \textsl{Spaces of {Q}uantum {F}ield {T}heories},
  \href{http://arxiv.org/abs/1005.2779}{[arXiv:1005.2779]}.

\bibitem{dgjt0203173}
M.~R. Douglas, S.~Govindarajan, T.~Jayaraman, and A.~Tomasiello, \textsl{D-branes
  on {C}alabi-{Y}au {M}anifolds and {S}uperpotentials}, Commun. Math. Phys.
  \textbf{248} (2004), 85--118,
  \href{http://arxiv.org/abs/hep-th/0203173}{[hep-th/0203173]}.

\bibitem{d0904.4713}
T.~Dyckerhoff, \textsl{Compact generators in categories of matrix factorizations},
  Duke Math. J. \textbf{159} (2011), 223--274,
  \href{http://arxiv.org/abs/0904.4713}{[arXiv:0904.4713]}.

\bibitem{dm1004.0687}
T.~Dyckerhoff and D.~Murfet, \textsl{The {K}apustin-{L}i formula revisited},
  \href{http://arxiv.org/abs/1004.0687}{[arXiv:1004.0687]}.

\bibitem{gj0608027}
S.~Govindarajan and H.~Jockers, \textsl{Effective superpotentials for {B}-branes
  in {L}andau-{G}inzburg models}, JHEP \textbf{0610} (2006), 060,
  \href{http://arxiv.org/abs/hep-th/0608027}{[hep-th/0608027]}.

\bibitem{GSW}
M.~B. Green, J.~H. Schwarz, and E.~Witten, \textsl{Superstring theory}, Cambridge
  Monographs on Mathematical Physics, Cambridge University Press, 1988.

\bibitem{GandH}
P.~Griffiths and J.~Harris, \textsl{Principles of algebraic geometry},
  Wiley-Interscience, 1994.

\bibitem{ghkk0909.2025}
T.~W. Grimm, T.-W. Ha, A.~Klemm, and D.~Klevers, \textsl{Computing {B}rane and
  {F}lux {S}uperpotentials in {F}-theory {C}ompactifications},
  \href{http://arxiv.org/abs/0909.2025}{[arXiv:0909.2025]}.

\bibitem{gkk1011.6375}
T.~W. Grimm, A.~Klemm, and D.~Klevers, \textsl{Five-brane superpotentials,
  blow-up geometries and {SU}(3) structure manifolds},
  \href{http://arxiv.org/abs/1011.6375}{[arXiv:1011.6375]}.

\bibitem{GuLaSt}
V.~K.~A.~M. Guggenheim, L.~A. Lambe, and J.~D. Stasheff, \textsl{Perturbation
  theory in differential homological algebra {II}}, Illinois J. Math.
  \textbf{35} (1991), 359--373.

\bibitem{hl0404}
M.~Herbst and C.~I. Lazaroiu, \textsl{Localization and traces in open-closed
  topological {L}andau-{G}inzburg models}, JHEP \textbf{0505} (2005), 044,
  \href{http://www.arxiv.org/abs/hep-th/0404184}{[hep-th/0404184]}.

\bibitem{hll0402}
M.~Herbst, C.~I. Lazaroiu, and W.~Lerche, \textsl{Superpotentials, ${A}_{\infty}$
  {R}elations and {WDVV} {E}quations for {O}pen {T}opological {S}trings}, JHEP
  \textbf{0502} (2005), 071,
  \href{http://www.arxiv.org/abs/hep-th/0402110}{[hep-th/0402110]}.

\bibitem{howewest3}
P.~Howe and P.~West, \textsl{Fixed points in multifield {L}andau-{G}insburg
  models}, Phys. Lett. B \textbf{244} (1990), 270--274.

\bibitem{js0808.0761}
H.~Jockers and M.~Soroush, \textsl{Effective superpotentials for compact
  {D5}-brane {C}alabi-{Y}au geometries}, Commun. Math. Phys. \textbf{290}
  (2009), 249--290, \href{http://arxiv.org/abs/0808.0761}{[arXiv:0808.0761]}.

\bibitem{k0504437}
T.~Kadeishvili, \textsl{On the homology theory of fibre spaces}, Uspekhi Mat.
  Nauk \textbf{35:3} (1980), 183--188, English translation in Russian
  Math.~Surveys, 35:3 (1980), 231--238,
  \href{http://arxiv.org/abs/math.AT/0504437}{[math.AT/0504437]}.

\bibitem{ks0410291}
H.~Kajiura and J.~Stasheff, \textsl{Homotopy algebras inspired by classical
  open-closed string field theory}, Commun. Math. Phys. \textbf{263} (2006),
  553--581, \href{http://arxiv.org/abs/math/0410291}{[math.QA/0410291]}.

\bibitem{ks0510118}
H.~Kajiura and J.~Stasheff, \textsl{Open-closed homotopy algebra in mathematical physics}, J. Math.
  Phys. \textbf{47} (2006), 023506,
  \href{http://arxiv.org/abs/hep-th/0510118}{[hep-th/0510118]}.

\bibitem{kl0210}
A.~Kapustin and Y.~Li, \textsl{D-branes in {L}andau-{G}inzburg {M}odels and
  {A}lgebraic {G}eometry}, JHEP \textbf{0312} (2003), 005,
  \href{http://www.arxiv.org/abs/hep-th/0210296}{[hep-th/0210296]}.

\bibitem{kl0305}
A.~Kapustin and Y.~Li, \textsl{Topological {C}orrelators in {L}andau-{G}inzburg {M}odels with
  {B}oundaries}, Adv. Theor. Math. Phys. \textbf{7} (2004), 727--749,
  \href{http://www.arxiv.org/abs/hep-th/0305136}{[hep-th/0305136]}.

\bibitem{kr0405232}
A.~Kapustin and L.~Rozansky, \textsl{On the relation between open and closed
  topological strings}, Commun. Math. Phys. \textbf{252} (2004), 393--414,
  \href{http://arxiv.org/abs/hep-th/0405232}{[hep-th/0405232]}.

\bibitem{Kassel}
C.~Kassel, \textsl{Homologie cyclique, caract\`ere de {C}hern et lemme de
  perturbation}, Journal f\"ur die reine und angewandte Mathematik \textbf{1990}
  (1990), 159--180.

\bibitem{kms1989}
D.~A. Kastor, E.~J. Martinec, and S.~H. Shenker, \textsl{{RG} {F}low in {$N=1$}
  {D}iscrete {S}eries}, Nucl. Phys. B \textbf{316} (1989), 590--608.

\bibitem{ks0805.1013}
J.~Knapp and E.~Scheidegger, \textsl{Towards {O}pen {S}tring {M}irror {S}ymmetry
  for {O}ne-{P}arameter {C}alabi-{Y}au {H}ypersurfaces},
  \href{http://arxiv.org/abs/0805.1013}{[arXiv:0805.1013]}.

\bibitem{kontsevich1993}
M.~Kontsevich, \textsl{Formal (non)commutative symplectic geometry}, The Gelfand
  Mathematical Seminars, 1990-1992, Fields Institute Communications,
  Birkh\"{a}user Boston, 1993.

\bibitem{k9709040}
M.~Kontsevich, \textsl{Deformation quantization of {P}oisson manifolds}, Lett. Math.
  Phys. \textbf{66} (2003), 157--216,
  \href{http://arxiv.org/abs/q-alg/9709040}{[math.QA/9709040]}.

\bibitem{KSbook}
M.~Kontsevich and Y.~Soibelman, \textsl{Deformation {T}heory. {I}}, book in
  preparation.

\bibitem{ks0606241}
M.~Kontsevich and Y.~Soibelman, \textsl{Notes on {A}-infinity algebras, {A}-infinity categories and
  non-commutative geometry. {I}},
  \href{http://arxiv.org/abs/math/0606241}{[math.RA/0606241]}.

\bibitem{kw0805.0792}
D.~Krefl and J.~Walcher, \textsl{Real {M}irror {S}ymmetry for {O}ne-parameter
  {H}ypersurfaces}, JHEP \textbf{0809} (2008), 031,
  \href{http://arxiv.org/abs/0805.0792}{[arXiv:0805.0792]}.

\bibitem{lm9406095}
T.~Lada and M.~Markl, \textsl{Strongly homotopy {L}ie algebras},
  \href{http://arxiv.org/abs/hep-th/9406095}{[hep-th/9406095]}.

\bibitem{l0204062}
A.~Lazarev, \textsl{Hochschild cohomology and moduli spaces of strongly homotopy
  associative algebras}, Homology Homotopy Appl. \textbf{5} (2003), 73--100,
  \href{http://arxiv.org/abs/math/0204062v3}{[math.QA/0204062]}.

\bibitem{l0010269}
C.~I. Lazaroiu, \textsl{On the structure of open-closed topological field theory
  in two dimensions}, Nucl. Phys.~B \textbf{603} (2001), 497--530,
  \href{http://www.arxiv.org/abs/hep-th/0010269}{[hep-th/0010269]}.

\bibitem{l0312}
C.~I. Lazaroiu, \textsl{On the boundary coupling of topological {L}andau-{G}inzburg
  models}, JHEP \textbf{0505} (2005), 037,
  \href{http://www.arxiv.org/abs/hep-th/0312286}{[hep-th/0312286]}.

\bibitem{l0507222}
C.~I. Lazaroiu, \textsl{On the non-commutative geometry of topological {D}-branes}, JHEP
  \textbf{0511} (2005), 032,
  \href{http://www.arxiv.org/abs/hep-th/0507222}{[hep-th/0507222]}.

\bibitem{lvw1989}
W.~Lerche, C.~Vafa, and N.~Warner, \textsl{Chiral {R}ings in {$N=2$}
  {S}uperconformal {T}heories}, Nucl. Phys.~B \textbf{324} (1989), 427.

\bibitem{Loday}
J.-L. Loday, \textsl{Cyclic homology}, Springer, 1997.

\bibitem{m1989}
E.~J. Martinec, \textsl{Algebraic {G}eometry and {E}ffective {L}agrangians},
  Phys. Lett.~B \textbf{217} (1989), 431.

\bibitem{m9809}
S.~A. Merkulov, \textsl{Strongly homotopy algebras of a {K}\"{a}hler manifold},
  Internat. Math. Res. Notices \textbf{3} (1999), 153--164,
  \href{http://arxiv.org/abs/math.AG/9809172}{[math.AG/9809172]}.

\bibitem{m0001007}
S.~A. Merkulov, \textsl{Frobenius$_\infty$ invariants of homotopy {G}erstenhaber
  algebras, {I}}, Duke Math. J. \textbf{105} (2000), 411--461,
  \href{http://arxiv.org/abs/math/0001007}{[math.AG/0001007]}.

\bibitem{ms0609042}
G.~W. Moore and G.~Segal, \textsl{D-branes and {K}-theory in {2D} topological
  field theory}, \href{http://arxiv.org/abs/hep-th/0609042}{[hep-th/0609042]}.

\bibitem{mw0709.4028}
D.~R. Morrison and J.~Walcher, \textsl{D-branes and {N}ormal {F}unctions},
  \href{http://arxiv.org/abs/0709.4028}{[arXiv:0709.4028]}.

\bibitem{m0912.1629}
D.~Murfet, \textsl{Residues and duality for singularity categories of isolated
  {G}orenstein singularities},
  \href{http://arxiv.org/abs/0912.1629}{[arXiv:0912.1629]}.

\bibitem{n0702449}
P.~Nicol\'as, \textsl{The bar derived category of a curved dg algebra}, Journal
  of Pure and Applied Algebra \textbf{212} (2008), 2633--2659,
  \href{http://arxiv.org/abs/math/0702449}{[math.RT/0702449]}.

\bibitem{pp1010.0982}
A.~Polishchuk and L.~Positselski, \textsl{Hochschild (co)homology of the second
  kind {I}}, \href{http://arxiv.org/abs/1010.0982}{[arXiv:1010.0982]}.

\bibitem{pv1002.2116}
A.~Polishchuk and A.~Vaintrob, \textsl{Chern characters and
  {H}irzebruch-{R}iemann-{R}och formula for matrix factorizations},
  \href{http://arxiv.org/abs/1002.2116}{[arXiv:1002.2116]}.

\bibitem{p0905.2621}
L.~Positselski, \textsl{Two kinds of derived categories, {K}oszul duality, and
  comodule-contramodule correspondence},
  \href{http://arxiv.org/abs/0905.2621}{[arXiv:0905.2621]}.

\bibitem{s0904.1339}
E.~Segal, \textsl{The closed state space of affine {L}andau-{G}inzburg
  {B}-models}, \href{http://arxiv.org/abs/0904.1339}{[arXiv:0904.1339v2]}.

\bibitem{t9803025}
D.~E. Tamarkin, \textsl{Another proof of {M}.~{K}ontsevich formality theorem},
  \href{http://arxiv.org/abs/math/9803025}{[math.QA/9803025]}.

\bibitem{v1991}
C.~Vafa, \textsl{Topological {L}andau-{G}inzburg {M}odels}, Mod. Phys. Lett.~A
  \textbf{6} (1991), 337--346.

\bibitem{vw1989}
C.~Vafa and N.~Warner, \textsl{Catastrophes and the classification of conformal
  theories}, Phys. Lett.~B \textbf{218} (1989), 51.

\bibitem{w0605162}
J.~Walcher, \textsl{Opening {M}irror {S}ymmetry on the {Q}uintic}, Commun. Math.
  Phys. \textbf{276} (2007), 671--689,
  \href{http://arxiv.org/abs/hep-th/0605162}{[hep-th/0605162]}.

\bibitem{w1988}
E.~Witten, \textsl{Topological {S}igma {M}odels}, Commun. Math. Phys.
  \textbf{118} (1988), 411.

\bibitem{zhoulecturenotes}
J.~Zhou, \textsl{Vertex operator algebras and differential geometry},  (2004),
  lecture notes.

\bibitem{z9206084}
B.~Zwiebach, \textsl{Closed string field theory: Quantum action and the {BV}
  master equation}, Nucl. Phys. B \textbf{390} (1993), 33--152,
  \href{http://www.arxiv.org/abs/hep-th/9206084}{[hep-th/9206084]}.

\bibitem{z9705241}
B.~Zwiebach, \textsl{Oriented open-closed string theory revisited}, Annals Phys.
  \textbf{267} (1998), 193--248,
  \href{http://www.arxiv.org/abs/hep-th/9705241}{[hep-th/9705241]}.

\end{thebibliography}
\end{document}